\documentclass[letterpaper, 10 pt, conference]{ieeeconf}  % Comment this line out
\IEEEoverridecommandlockouts                              % This command is only
\usepackage{amsmath}
\usepackage{amssymb}
\usepackage{mathtools}
\usepackage{algorithm}
\usepackage{algorithmic}

\newcommand{\lp}{\left(}
\newcommand{\rp}{\right)}
\newcommand{\lb}{\left[}
\newcommand{\rb}{\right]}
\newcommand{\lbp}{\left\{}
\newcommand{\rbp}{\right\}}

\newcommand{\mcal}{\mathcal}
\newcommand{\mrm}{\mathrm}

\newcommand{\mbb}{\mathbb}

\newcommand{\lce}{\left\lceil}
\newcommand{\rce}{\right\rceil}

\newcommand{\E}{\mathbb{E}}

\newcommand{\argmin}{\mathop{\mathrm{argmin}}}
\newcommand{\argmax}{\mathop{\mathrm{argmax}}}

\makeatletter
\newcommand*{\indep}{%
  \mathbin{%
    \mathpalette{\@indep}{}%
  }%
}
\newcommand*{\nindep}{%
  \mathbin{%                   % The final symbol is a binary math operator
    \mathpalette{\@indep}{\not}% \mathpalette helps for the adaptation
  }%
}
\newcommand*{\@indep}[2]{%
  \sbox0{$#1\perp\m@th$}%        box 0 contains \perp symbol
  \sbox2{$#1=$}%                 box 2 for the height of =
  \sbox4{$#1\vcenter{}$}%        box 4 for the height of the math axis
  \rlap{\copy0}%                 first \perp
  \dimen@=\dimexpr\ht2-\ht4-.2pt\relax
  \kern\dimen@
  {#2}%
  \kern\dimen@
  \copy0 %                       second \perp
} 
\makeatother
\title{\LARGE \bf
Learning Where to Look: UCB-Driven Controlled Sensing for Quickest Change Detection
}

\author{Yu-Han Huang, Argyrios Gerogiannis, Subhonmesh Bose, Venugopal V. Veeravalli}

\begin{document}

\maketitle
\thispagestyle{empty}
\pagestyle{empty}

%%%%%%%%%%%%%%%%%%%%%%%%%%%%%%%%%%%%%%%%%%%%%%%%%%%%%%%%%%%%%%%%%%%%%%%%%%%%%%%%
\begin{abstract}
%Many modern systems rely on large sets of sensors to detect abrupt changes in underlying dynamics. Such changes are often localized: only a small and unknown subset of sensors becomes informative after the change, rendering uniform monitoring both inefficient and costly. To this end, 

We study the multichannel quickest change detection problem with bandit feedback and controlled sensing, in which an agent sequentially selects one of the data streams to observe at each time-step and aims to detect an unknown change as quickly as possible while controlling false alarms. Assuming known pre- and post-change distributions and allowing an arbitrary subset of streams to be affected by the change, we propose two novel and computationally efficient detection procedures inspired by the Upper Confidence Bound (UCB) multi-armed bandit algorithm. Our methods adaptively concentrate sensing on the most informative streams while preserving false-alarm guarantees. We show that both procedures achieve first-order asymptotic optimality in detection delay under standard false-alarm constraints. We also extend the UCB-driven controlled sensing approach to the setting where the pre- and post-change distributions are unknown, except for a  mean-shift in at least one of the channels at the change-point.  This setting is particularly relevant to the problem of learning in piecewise stationary environments. Finally, extensive simulations on synthetic benchmarks show that our methods consistently outperform existing state-of-the-art approaches while offering substantial computational savings.

  %For learning in piecewise-stationary environments, changes in environments affects observations from actions or control policies differently. This work provides a step towards a principled forced exploration algorithm for change detection in such environments.
  
\end{abstract}

\section{Introduction}

The problem of Quickest Change Detection (QCD) has an extensive range of applications in science and engineering, including quality control \cite{hawkins2003changepoint}, earthquake detection \cite{zhou2025sequential}, and online learning \cite{huang2025sequential}. In this problem, an agent sequentially observes the environment through noisy samples, whose distribution shifts when the environment undergoes a change. The goal for the agent is to detect the change as soon as possible while not triggering false alarms too often. See \cite{poor-hadj-qcd-book-2009, vvv_qcd_overview, tart-niki-bass-2014, xie_vvv_qcd_overview} for books and survey articles on this topic.

%\argy{Here it would make sense to practically motivate the bandit-qcd problem saying that in many applications, sensors exist that monitor a procedure (e.g. in production or security) but when a change happens only a subset of the sensors is informative and we wanna be able to focus to the said change as fast as possible.}

In this work, we study a variant of QCD, referred to as \emph{multichannel bandit QCD}  \cite{xu2021optimum,zhang2023bandit,xu2023asymptotic}, which is a special case of the more general problem of QCD with controlled sensing that was introduced in \cite{gopalan2021bandit} and further studied in \cite{veeravalli2024quickest}. In the multichannel bandit QCD problem, there are a finite number of data streams, and each data stream undergoes a different degree of shift in its distribution when the change occurs. The agent can only sample from \emph{one} stream at a time, and therefore requires a carefully designed control policy that selects which data stream to observe. In particular, the change detector and the control policy have to be designed jointly to achieve the best tradeoff between detection delay and false alarm rate. 

A canonical example application of the QCD problem with controlled sensing, described in \cite{gopalan2021bandit}, is the surveillance system with a limited number of sensors that can be steered in different directions and locations. In this scenario, only a few sensors with specific positions can detect the change, as changes are localized in certain areas. 
%This type of surveillance system has a wide range of applications, including anomaly detection in online social networks \cite{viswanath2014towards}, quality control in manufacturing systems \cite{ding2006distributed}, and intrusion detection in computer networks \cite{bass1999multisensor}. 

%Additional examples are discussed in \cite{gopalan2021bandit,zhang2023bandit}. 

Our motivation for studying the multichannel bandit QCD problem comes from detection-augmented learning in piecewise stationary environments \cite{huang2025detection,gerogiannis2025detection}, where one detects changes in the environment to restart stationary  learning algorithms. To detect such changes, one cannot rely on the optimal learned algorithm over a stationary interval to detect changes in the environment quickly; it requires forced exploration. In this context, each channel  corresponds to an action chosen for forced exploration, with the observations being the rewards obtained from that action. When the learning environment changes, this change is reflected in the reward distributions obtained from the exploration actions, but the degree to which the distribution changes can vary  across different actions. Current strategies for change detection typically cycle through the various exploration actions in a round-robin fashion \cite{besson2022efficient,huang2025detection,gerogiannis2025detection}. The round-robin approach is wasteful if only a small subset of the actions see significant shifts in distribution. Ideally, one would like to only explore actions for which the distribution shift is the most pronounced, making it a multichannel bandit QCD problem instance.

%\subsection{Existing Algorithms for Multichannel Bandit QCD}
A greedy algorithm for the multichannel bandit QCD problem was proposed in \cite{xu2021optimum} that samples a single channel until a cumulative log-likelihood 
ratio statistic for that channel exceeds a certain positive threshold or falls below zero. In the former case, a change is declared, and in the latter, the algorithm discards the observations from the channel, and repeats the process on the next channel. This greedy algorithm was shown to possess certain asymptotic optimality properties in terms of the tradeoff between delay delay and the false alarm rate when only a single channel is affected by the change or all affected channels see identical distribution shifts. This algorithm can be far from optimal when there is more than one affected channel and the distribution shifts can vary significantly across actions. Intuitively, the greedy algorithm can focus on a single channel with a small distribution shift and declare a change with a large delay, never getting to sample another channel with a larger distribution shift, capable of certifying the change with far fewer samples. 
To circumvent this challenge, the authors of \cite{gopalan2021bandit} proposed the $\epsilon$-Greedy Change Detector ($\epsilon$-GCD) for the general problem of QCD with controlled sensing. This algorithm, however, is
% An alternative algorithm that aims to circumvent this shortcoming of the greedy algorithm was proposed in \cite{gopalan2021bandit}, in which the authors address the general problem of QCD with controlled sensing. However, the $\epsilon$-Greedy Change Detector ($\epsilon$-GCD) developed in \cite{gopalan2021bandit} is also 
not (asymptotically) optimal; the worst-case delay can increase with the pre-change duration, as pointed out by \cite{veeravalli2024quickest}. 

An asymptotically optimal procedure for the general problem of QCD with controlled sensing, called the Windowed-Chernoff-Cumulative Sum (WCC) procedure, was provided in \cite{veeravalli2024quickest}. In the special case of the multichannel bandit QCD problem, the WCC procedure chooses which channel to sample based on a windowed maximum-likelihood estimate (MLE) of the set of channels that are affected by the channel, and then chooses the channel with the largest potential distribution shift within that set.  In addition, to avoid being unduly biased by the MLE, the WCC procedures performs random exploration of the other channels at $o(n)$ time-steps over a time-horizon of length $n$. The observations from the chosen channels at each time-step are then fed into a windowed Cumulative Sum (CuSum) test for detecting the change.
%and a change is declared when the CuSum statistic exceeds a certain threshold.

While the WCC procedure is asymptotically optimal for the multichannel bandit QCD problem, the control policy, which is designed to work for the general problem of QCD with controlled sensing, is not well-matched to the multichannel special case. In particular, the number of subsets over which the MLE is computed grows exponentially with the number of channels. Another aspect of the WCC procedure that might make it less useful in the multichannel setting is that the observations across channels are combined into a single windowed CuSum test for change detection. As discussed in Section~\ref{sec:unknown}, this aspect does not allow the WCC procedure to be easily generalized to the setting where the pre- and post-change distributions of the observations from the channels are unknown.

%In this paper, we develop new procedures for the multichannel bandit QCD problem that are inspired by the 

%Although this procedure is shown to achieve first-order asymptotic optimality, its control policy does not follow classical algorithms in bandit literature, which learn the best action efficiently without exploring other suboptimal ones excessively. Therefore, it is natural to believe that employing a bandit algorithm as the control policy will improve the performance.

%In addition, WCC procedure combines the LLRs from different data streams into a single statistic. 

%As shown later in Section \ref{sec:unknown}, it is not straightforward to generalize this statistic to scenarios in which the distributions are unknown to the agent, as the observations from different streams follow different distributions. Consequently, it is desirable to separately design a statistic for each data stream.

Given the similarity between the multi-armed bandit (MAB) and the multichannel bandit QCD problems, it is natural to ask if one can develop MAB-like control policies for the latter problem. In order to do so, we need to identify an appropriate (pseudo) reward that can be used to guide action selection for change detection. Since the channel with the largest distribution shift (measured in terms of the Kullback-Leibler (KL) divergence) will generate observations that have the largest average log-likelihood ratio (LLR) after the change, we may use LLR of the observations as the reward. We design efficient procedures for the multichannel bandit QCD problem by combining the popular Upper Confidence Bound (UCB) MAB algorithm with CuSum change detection by viewing LLRs as rewards from the actions. 

% \bose{Here, say... one can design a statistic that combines LLR from all arms and name the algorithm. Then, say the problem with WCC and the UCB-CuSum procedure that they do not easily generalize to the case when the distributions are not known. Then say, we study PA-UCB-CuSum.}

Concretely, our contributions are as follows: \textit{(i)} We propose \textbf{UCB-CuSum}, a bandit QCD procedure using UCB-based action selection to compute a CuSum statistic for change detection that is proven to be asymptotically optimal. \textit{(ii)}     We generalize the above algorithm to \textbf{PA-UCB-CuSum} that maintains a separate CuSum statistic for each action for change detection, which is shown to offer a straightforward extension to settings with unknown pre- and post-change distributions. \textit{(iii)} The proposed methods are empirically shown to outperform prior methods.

\section{Problem Formulation}
\label{sec:problem_form}

In this section, we formally define a multichannel bandit QCD problem. We start by defining some notations. Let $\lb n \rb \coloneqq \lbp 1, \dots, n \rbp$ for any $n \in \mbb{N}$, and 
$\left\{ X_{a,n} : a \in [K], n \in \mbb{N}\right\}$ be $K$ sequences of mutually independent stochastic observations indexed by the action $a$ and the time-step $n$. For each action $a$, there are two densities the observation $X_{a,n}$ can possibly follow: the pre-change density $f_{a,0}$ and the post-change density $f_{a,1}$. At some time-step $\nu$, referred to as the \emph{change-point}, the underlying distributions of the stochastic observations could possibly change, i.e., for $a \in \lb K \rb$ and $n \in \mbb{N}$,
\begin{equation}\label{eq:obs_distr}
    X_{a,n} \sim \begin{dcases}
        f_{a,1},& n \geq \nu \mrm{\: and\:} a \in \mcal{A}, \\
        f_{a,0}, & \mrm{otherwise},
    \end{dcases}
\end{equation}
where $\mcal{A} \subseteq \lb K \rb$ is an nonempty subset that denotes the actions that undergo a distributional shift in the observations. In other words, for action $a$ in the subset $\mcal{A}$, the observations follows the pre-change density $f_{a,0}$ before the change-point $\nu$ and follows the post-change density $f_{a,1}$ after the change-point $\nu$. For actions that are not in the subset $\mcal{A}$, the observations follows $f_{a,0}$ before and after $\nu$, as the change does not affect the stochastic properties associated with these actions. The densities $f_{a,0}$ and $f_{a,1}$ are with respect to the same dominating measure $\lambda$, and they are all fully known by the agent.\footnote{In our experiments, we explore cases where the pre- and post-change densities are bounded but not known.} However, the change-point $\nu$ and the subset $\mcal{A}$ are deterministic and completely unknown to the agent. We also make the following assumptions, which does not impose a strict restriction on the pre- and post-change densities.

\begin{assumption}\label{assum:KL}
Assume that the KL Divergence between $f_{a,1}$ and $f_{a,0}$ is positive, i.e., for any $a \in \lb K \rb$,
\begin{equation}%\label{eq:KL}
    D \lp f_{a,1} || f_{a,0} \rp \coloneqq \int \log \lp f_{a,1} / f_{a,0} \rp f_{a,1}d\lambda > 0.
\end{equation}
In addition, assume that there exists a constant $v < \infty$ such that for any $a \in \lb K \rb$, $\log\lp f_{a,1} / f_{a,0} \rp$ is $v$-sub-Gaussian, i.e., for any $\theta \in \mbb{R}$,
\begin{align}%\label{eq:sec_ord_assum}
    &\log \lp \int \exp \lp \theta \lp \log \lp f_{a,1} / f_{a,0} \rp - D \lp f_{a,1} || f_{a,0} \rp \rp \rp f_{a,1}d\lambda \rp \nonumber\\
    &\leq \frac{v}{2}\theta^{2}.
\end{align}
\end{assumption}
\vspace{2pt}

At time-step $n$, the agent chooses an action $A_{n}$ to determine which data stream to observe. These actions are determined based on past observations; in particular, $A_{n}$ is $\mcal{F}_{n-1}$ measurable where $\mcal{F}_{n-1} \coloneqq \sigma \lp X_{A_{1},1}, \dots, X_{A_{n-1},n-1} \rp$ is the filtration generated by all observations collected prior to time-step $n$. The sequence of actions $A \coloneqq \lbp A_{n}: n \in \mbb{N} \rbp$ forms a \emph{control policy}. 
Based on the noisy observations obtained from the control policy, the agent decides whether a change has occurred or not. The stopping time at which the agent detects a change is denoted by $T$, and we refer to the pair $\lp A, T \rp$ as the \emph{bandit QCD procedure}, as these two components determine the detection delay and the false alarm frequency.

We use $\mbb{P}_{\nu, \mcal{A}}^{A}$ and $\E_{\nu, \mcal{A}}^{A}$ to denote the  measure and expectation when the agent uses $A$ as its control policy, and the density associated with actions in $\mcal{A}$ changes at $\nu$. We also use $\mbb{P}_{\infty}^{A}$ and $\E_{\infty}^{A}$ to denote those when the agent uses $A$ as its control policy and no changes occur. For the false alarm and detection delay metrics, we employ Lorden's criterion \cite{lorden1971procedures}. The false alarm is evaluated by the mean time to false alarm $\E_{\infty}^{A}\lb T \rb$,\footnote{This false alarm metric is different from but related to the false alarm probability used in many prior works.} and the detection delay is evaluated by the worst-case, over all possible change-points and pre-change observations, expected detection delay conditioned on the filtration $\mcal{F}_{\nu-1}$, i.e., for any control policy $A$ and $\mcal{A}\subseteq \lb K \rb$,
\begin{equation}\label{eq:WADD}
    \!\!\mcal{J}_{\mcal{A}}\lp A, T \rp\!\coloneqq\! \sup_{\nu \in \mbb{N}}\mrm{\:ess\:}\sup \E_{\nu,\mcal{A}}^{A} \!\lb \lp T - \nu + 1\rp^{+} \!\Big| \mcal{F}_{\nu-1} \rb
\end{equation}
where $\lp x \rp^{+} \coloneqq \max\{x, 0\}$ is the positive part of $x \in \mbb{R}$. The definition in \eqref{eq:WADD} follows from the delay measure in \cite{veeravalli2024quickest}. The goal of the agent is to minimize the worst-case expected delay $\mcal{J}_{\mcal{A}}^{A}$, while controlling the mean time to false alarm to be above a predetermined level $\gamma$, i.e.,
\begin{equation}\label{eq:QCD}
    \min_{T,A} \mcal{J}_{\mcal{A}}\lp A, T \rp, \:\textrm{subject to}\:\E_{\infty}^{A}\lb T \rb \geq \gamma.
\end{equation}
\section{Our Procedures}

\subsection{UCB-CuSum Procedure}
In the QCD problem under Lorden's criterion without control sensing, i.e., $K=1$ and $\mcal{A} = \lbp 1 \rbp$, the CuSum test \cite{page1954continuous} is asymptotically optimal as $\gamma \to \infty$ \cite{lorden1971procedures, lai1998information}. The CuSum statistic is computed by a recursion that adds the LLR of the observation at time-step $n$ with the positive part of the CuSum statistic at the previous time-step $n-1$. The CuSum test raises an alarm whenever the CuSum statistic surpasses a constant threshold. Due to the asymptotic optimality, it is natural to expect that a procedure that utilizes a CuSum-like statistic will perform well in the bandit QCD problem. Thus, following the recursion along the lines of \cite{zhang2023bandit,veeravalli2024quickest}, we define a CuSum-like statistic: for $n \in \mbb{N}$,
\begin{equation}\label{eq:CuSum_stat}
    C_{n}^{A} \coloneqq \lp C_{n-1}^{A} \rp^{+} + \mrm{LLR} \lp A_{n}, X_{A_{n}, n} \rp
\end{equation}
with $C_{0}^{A} \coloneqq 0$ and the LLR function
\begin{equation}\label{eq:LLR}
    \mrm{LLR} \lp a, x \rp \coloneqq \log \lp f_{a, 1}\lp x \rp/f_{a, 0} \lp x \rp \rp.
\end{equation}
The procedure using the CuSum-like statistic in \eqref{eq:CuSum_stat} stops whenever the statistic surpasses a constant threshold $b$, i.e.,
\begin{equation}\label{eq:CuSum_T}
    T_{b}^{A} \coloneqq \inf \lbp n \in \mbb{N}: C_{n}^{A} \geq b \rbp.
\end{equation}
We choose $b$ to satisfy the false alarm constraint in \eqref{eq:QCD}.

To achieve faster detection, the control policy $A$ should select the action with the largest average LLR after the change-point, so that the CuSum-like statistic in \eqref{eq:CuSum_stat} grows quickly to surpass the threshold. Therefore, to learn the optimal action that leads to the largest increment in the statistic in \eqref{eq:CuSum_stat}, we employ an upper confidence bound (UCB) algorithm \cite{auer2002finite} with periodic restarts $A_{\mrm{UCB}}$ as our control policy $A$. First, we partition the entire horizon into intervals of length $W$, i.e., $\mbb{N} = \cup_{j=1}^{\infty} \mcal{I}_{j}$ with the $j^{\mrm{th}}$ interval $\mcal{I}_{j} \coloneqq \lbp \lp j-1 \rp W + 1, \dots, jW \rbp$. In each interval, we apply the UCB algorithm to select the action and restart the control policy at the start of the next interval. 
Notice that we restart our UCB algorithm every $W$ time-steps. Without such a restart, the control policy tends to heavily concentrate on the action with the largest average LLR before the change. Consequently, the control policy rarely explore other actions, leading to delayed detection after the change.
We view the LLR of the observation  $X_{A_{n}, n}$ from action $A_{n}$ as the reward at time-step $n$, that is,
\begin{equation}\label{eq:reward}
    R_{n} \coloneqq \mrm{LLR} \lp A_{n}, X_{A_{n}, n} \rp.
\end{equation}
Let $j = \lce n / W \rce$ denote the index of the current interval, and let $m = \lp j - 1 \rp W + 1$ denote the start of the current interval. In addition, we use $N_{a, n}$ to denote the times action $a$ is selected since the start of the current interval, i.e.,
\begin{align}\label{eq:N_an}
    N_{a, n} \coloneqq \sum_{i=m}^{n} \mbb{I} \lbp A_{i} = a \rbp.
\end{align}
Let $\hat{\mu}_{a, n}$ denote the empirical mean of the rewards from action $a$ since the start of the current interval, i.e.,
\begin{align}\label{eq:mu_hat}
    \hat{\mu}_{a, n} \coloneqq \begin{dcases}
        \frac{1}{N_{a,n}} \sum_{i=m}^{n} R_{i}\mbb{I} \lbp A_{i} = a \rbp,& N_{a,n} > 0\\
        0,& N_{a,n} = 0.
    \end{dcases}
\end{align}
Then, the UCB index is defined as follows:
\begin{equation}\label{eq:UCB_index}
    \mrm{UCB} \lp a, n \rp \coloneqq \hat{\mu}_{a, n} + \sqrt{\frac{4v\log W}{N_{a,n}}},
\end{equation}
and the UCB algorithm selects the action with the largest UCB index, i.e., $A_{\mrm{UCB}} \coloneqq \lbp A_{n}: n \in \mbb{N} \rbp$ with
\begin{equation}\label{eq:UCB_Actions}
    A_{n} \coloneqq \argmax_{a \in \lb K \rb} \mrm{UCB} \lp a, n \rp.
\end{equation}
Combining the UCB control policy $A_{\mrm{UCB}}$ in \eqref{eq:UCB_Actions} with the stopping time $T_{b}^{A}$ in \eqref{eq:CuSum_T}, we propose the UCB-CuSum Procedure $\lp A_{\mrm{UCB}}, T_{b}^{A_{\mrm{UCB}}} \rp$, which is illustrated in Algorithm \ref{alg:UCB_CuSum}. This procedure uses the CuSum-like statistic in \eqref{eq:CuSum_stat} that combines LLRs from all actions together. 

\begin{algorithm}[htb]
  \caption{UCB-CuSum Procedure}
  \label{alg:UCB_CuSum}
  \begin{algorithmic}
    \STATE {\bfseries Input:} Threshold $b$, Interval length $W$
    \STATE Initialize $C \leftarrow 0$ and $n \leftarrow 0$.
    \REPEAT
    \IF{$n \mod W = 0$}
    \STATE Set $N_{a} \leftarrow 0$, $\hat{\mu}_{a} \leftarrow 0$, $\mrm{UCB}_{a} \leftarrow \infty$ for $a \in \lb K \rb$.
    \ENDIF
    \STATE Choose action $A \in \argmax_{a \in \lb K \rb} \mrm{UCB}_{a}$.
    \STATE Receive observation $X$.
    \STATE Update statistic $C \leftarrow \max\lbp C, 0 \rbp + \mrm{LLR}\lp A, X \rp$.
    \STATE Compute reward $R \leftarrow \mrm{LLR} \lp A, X \rp$.
    \STATE Update mean reward $\hat{\mu}_{A} \leftarrow \lp N_{a} \mu + R \rp/\lp N_{a} + 1 \rp$.
    \STATE Update index $\mrm{UCB}_{A} \leftarrow \hat{\mu}_{A} + \sqrt{4v\log W/N_{A}}$.
    \STATE Update number of pulls $N_{A} \leftarrow N_{A} + 1$.
    \STATE $n \leftarrow n + 1$.
    \UNTIL{$C \geq b$}
  \end{algorithmic}
\end{algorithm}

In the following result, we characterize the performance of the UCB-CuSum procedure.
% show that with a specific choice of the threshold $b$ and the interval length $W$, the UCB-CuSum procedure satisfies the mean time to false alarm constraint in \eqref{eq:QCD} and achieves first-order asymptotic optimality as $\gamma$ goes to infinity.
%
\begin{theorem}[Asymptotic optimality of UCB-CuSum]\label{thm:UCB_CuSum_opt}
With $b = \log \gamma$, $\E_{\infty}^{A_{\mrm{UCB}}} \lb T_{b}^{A_{\mrm{UCB}}} \rb \geq \gamma$. As $\gamma, W \to \infty$ with $W = o\lp \log \gamma \rp$, we have
\begin{equation}\label{eq:thm_WADD} \mcal{J}_{\mcal{A}} \!\lp \!A_{\mrm{UCB}}, T_{b}^{A_{\mrm{UCB}}} \!\rp \leq \frac{\log \gamma}{I_{\mcal{A}}} \lp 1 + o \lp 1\rp \rp.
\end{equation}
where $I_{\mcal{A}} \coloneqq \max_{a \in \mcal{A}} D\lp f_{a,1} || f_{a, 0} \rp$.
\end{theorem}
The formal proof is included in Appendix \ref{sec:thm1}. Here, we provide a sketch. 
The lower bound on the mean time to false alarm follows from Lemma 1 in \cite{veeravalli2024quickest}, which relies on the martingale property of a Shiryaev-Roberts (SR)-like statistic and the optional sampling theorem.
To prove the upper bound on the expected delay, we establish that 
% There are two main steps to this upper bound on the worst-case expected delay $\mcal{J}_{\mcal{A}}$. First, we show that
%
\begin{equation}
    \mcal{J}_{\mcal{A}} \lp A_{\mrm{UCB}}, T_{b}^{A_{\mrm{UCB}}} \rp \leq W + \E_{1,\mcal{A}}^{A_{\mrm{UCB}}} \lb T_{b}^{A_{\mrm{UCB}}} \rb,
\end{equation}
where recall that $\E_{1,\mcal{A}}^{A_{\mrm{UCB}}}$ denotes the expectation when the change occurs at time-step $1$.
The intuition behind this step is that the duration between the change-point $\nu$ and the subsequent restart $\ell$ is at most $W$, and that $\E_{\nu,\mcal{A}}^{A_{\mrm{UCB}}} [T_{b}^{A_{\mrm{UCB}}} - \ell + 1 | \mcal{F}_{\nu-1}]$ is roughly equal to $\E_{1,\mcal{A}}^{A_{\mrm{UCB}}} [T_{b}^{A_{\mrm{UCB}}}]$, as the observations after the restart $\ell$ under $\mbb{P}_{\nu,\mcal{A}}^{A_{\mrm{UCB}}}$ and all observations under $\mbb{P}_{1,\mcal{A}}^{A_{\mrm{UCB}}}$ follow the same post-change distribution. Next, we upper bound $\E_{1,\mcal{A}}^{A_{\mrm{UCB}}} [T_{b}^{A_{\mrm{UCB}}} ]$. While our precise argument relies on
Proposition 1 in \cite{veeravalli2024quickest}, we provide some intuition here. To do so, consider the notation,
\begin{gather}\label{eq:Sum_LLR}
    Y_{i}^{A_{\mrm{UCB}}} \coloneqq \sum_{j=(i-1)W + 1}^{iW} \mrm{LLR}\lp A_{j}, X_{A_{j}, j}\rp, i \in \mbb{N}.
\end{gather}
Thus, $Y_i$ is the sum of the LLR over the $i$-th $W$-step window. Let $\bar{C}_{iW}^{A_{\mrm{UCB}}} := \sum_{j=1}^{iW} Y_{j}^{A_{\mrm{UCB}}}$ accumulate these $Y_{j}^{A_{\mrm{UCB}}}$'s across windows up to the $i$-th one. With this notation, notice that 
\begin{equation}
    T_{b}^{A_{\mrm{UCB}}} \leq W \inf \lbp i \in \mbb{N}: \bar{C}_{iW}^{A_{\mrm{UCB}}} \geq b \rbp
\end{equation}
%
% Since $C_{n}^{A_{\mrm{UCB}}} \geq \bar{C}_{n}^{A_{\mrm{UCB}}}$ for any $n \in \mbb{N}$, $T_{b}^{A_{\mrm{UCB}}}$ is upper bounded by $T'$. 
% Then, by Wald's identity, the expected value of $T'$ is approximately $bW$ divided by the expected value of $Y_{i}^{A_{\mrm{UCB}}}$'s.
To obtain the expected hitting time of $\bar{C}^{A_{\mrm{UCB}}}_{iW}$'s, we view the LLR's as rewards for an UCB algorithm in \eqref{eq:UCB_Actions} for which the expectation of $Y_i^{A_{\mrm{UCB}}}$, given the history $\mcal{F}_{(i-1)W}$, precisely equals the maximum possible cumulative reward $WI_{\mcal{A}}$ less the regret of the UCB algorithm. We borrow an upper bound on that regret from Theorem 7.1 in  \cite{lattimore2020bandit}. In turn, this step provides a minimum expected positive drift on $Y_{i}^{A_{\mrm{UCB}}}$'s using regret analysis of the UCB algorithm, with which we apply the level crossing result in Proposition 1 for $\bar{C}^{A_{\mrm{UCB}}}_{iW}$ from \cite{veeravalli2024quickest}. Putting it together, we obtain
\begin{equation}
    \!\!\E_{1,\mcal{A}}^{A_{\mrm{UCB}}} \!\!\lb T_{b}^{A_{\mrm{UCB}}} \rb \!\leq\! \frac{W \lp b + c \lp 1 + \sqrt{b} \rp \rp}{WI_{\mcal{A}} -  \sum_{a: \Delta_{a} > 0} \lp 3\Delta_{a} + \frac{16 \log W}{\Delta_{a}} \rp}.\nonumber
\end{equation}
The rest follows from taking $\gamma \to \infty$ with $W = o(\log \gamma)$.

UCB-CuSum is in fact asymptotically optimal in the regime $\gamma \to \infty$; optimality follows from a known lower bound on the worst-case expected delay $\mcal{J}_{\mcal{A}}$ of any multichannel bandit QCD procedure that satisfies the false alarm constraint. More precisely, as $\gamma \to \infty$,
\begin{equation} \label{eq:WADD_low}
\inf_{\lp A, T \rp: \E_{\infty}^{A} \lb T \rb \geq \gamma} \mcal{J}_{\mcal{A}} \lp A, T \rp \geq \frac{\log \gamma}{I_{\mcal{A}}} \lp 1 + o \lp 1 \rp \rp,
\end{equation}
according to Theorem 1 in \cite{veeravalli2024quickest}.

\subsection{Per-Action-UCB-CuSum Procedure}

The CuSum-like statistic in \eqref{eq:CuSum_stat} accumulates LLRs associated with different pairs of pre- and post-change distributions. This is different from the original CuSum statistic, which only involves a single pair of pre- and post-change distributions. Now, we study an alternate CuSum statistic that is defined by the following recursion: for any $a \in \lb K \rb$ and $n \in \mbb{N}$,
\begin{align}\label{eq:CuSum_stat_sep}
    C_{a,n}^{A} \coloneqq \begin{dcases}
        \lp C_{a, n-1}^{A} \rp^{+} + \mrm{LLR} \lp a, X_{a, n} \rp,& A_{n} = a\\
        C_{a, n-1}^{A},& \mrm{otherwise}
    \end{dcases}
\end{align}
with $C^{A}_{a,0} \coloneqq 0$ and $\mrm{LLR}$ defined in \eqref{eq:LLR}. The change detection procedure using the statistic in \eqref{eq:CuSum_stat_sep} flags a change whenever one of these statistics surpasses a constant threshold $b$, i.e.,
\begin{equation}\label{eq:CuSum_T_sep}
    \tilde{T}_{b}^{A} \coloneqq \inf \lbp n \in \mbb{N}: C_{a,n}^{A} \geq b \mrm{\:for\:some\:} a \in \lb K \rb \rbp.
\end{equation}
Combining with the UCB algorithm, we propose the Per-Action (PA)-UCB-CuSum procedure $\lp A_{\mrm{UCB}}, \tilde{T}_{b}^{A_{\mrm{UCB}}} \rp$, which computes a CuSum statistic in \eqref{eq:CuSum_stat_sep} for each action separately. The PA-UCB-CuSum procedure is illustrated in Algorithm \ref{alg:PA_UCB_CuSum}. This algorithm is motivated by the fact that it generalizes better to the case when the pre- and post-change distributions are not known, and the statistics must be computed using generalized LLRs, as we illustrate later in Section \ref{sec:unknown}.

\begin{algorithm}[htb]
  \caption{PA-UCB-CuSum Procedure}
  \label{alg:PA_UCB_CuSum}
  \begin{algorithmic}
    \STATE {\bfseries Input:} Threshold $b$, Interval length $W$
    \STATE Initialize $n \leftarrow 0$ and $C_{a} \leftarrow 0$ for all $a \in \lb K \rb$.
    \REPEAT
    \IF{$n \mod W = 0$}
    \STATE Set $N_{a} \leftarrow 0$, $\hat{\mu}_{a} \leftarrow 0$, $\mrm{UCB}_{a} \leftarrow \infty$ for $a \in \lb K \rb$.
    \ENDIF
    \STATE Choose action $A \in \argmax_{a \in \lb K \rb} \mrm{UCB}_{a}$.
    \STATE Receive observation $X$.
    \STATE Update statistic $C_{A} \leftarrow \max\lbp C_{A}, 0 \rbp + \mrm{LLR}\lp A, X \rp$.
    \STATE Compute reward $R \leftarrow \mrm{LLR} \lp A, X \rp$.
    \STATE Update mean reward $\hat{\mu}_{A} \leftarrow \lp N_{a} \mu + R \rp/\lp N_{a} + 1 \rp$.
    \STATE Update index $\mrm{UCB}_{A} \leftarrow \hat{\mu}_{A} + \sqrt{4v\log W/N_{A}}$.
    \STATE Update number of pulls $N_{A} \leftarrow N_{A} + 1$.
    \STATE $n \leftarrow n + 1$.
    \UNTIL{$C_{a} \geq b$ for some $a \in \lb K \rb$}
  \end{algorithmic}
\end{algorithm}

Our next result resonates the same message as in Theorem \ref{thm:UCB_CuSum_opt} for the PA-UCB-CuSum procedure, claiming its asymptotic optimality under the false alarm constraint in \eqref{eq:QCD}.
\begin{theorem}[Asymptotic optimality of PA-UCB-CuSum]\label{thm:PA_UCB_CuSum_opt}
With $b = \log \gamma$, $\E_{\infty}^{A_{\mrm{UCB}}} \lb \tilde{T}_{b}^{A_{\mrm{UCB}}} \rb \geq \gamma$ and as $\gamma, W \to \infty$ with $W = o\lp \log \gamma \rp$, we have
\begin{equation}\label{eq:thm_PA_WADD}
\mcal{J}_{\mcal{A}} \!\lp \!A_{\mrm{UCB}}, \tilde{T}_{b}^{A_{\mrm{UCB}}} \!\rp \leq \frac{\log\gamma}{I_{\mcal{A}}} \lp 1 + o \lp 1 \rp \rp.
\end{equation}
\end{theorem}
\vspace{3pt}
% \begin{proof}[Proof Sketch]
The proof of the upper bound on the worst-case expected delay $\mcal{J}_{\mcal{A}} (A_{\mrm{UCB}}, \tilde{T}^{A_{\mrm{UCB}}}_{b})$ follows similar steps as in the proof of Theorem \ref{thm:UCB_CuSum_opt}. However, we cannot apply Lemma 1 in \cite{veeravalli2024quickest} to argue the lower bound on the false alarm constraint. 
% that $\E_{\infty}^{A_{\mrm{UCB}}} \lb \tilde{T}_{b}^{A_{\mrm{UCB}}} \rb \geq \gamma$, since the CuSum-like statistic in \cite{veeravalli2024quickest} accumulates LLRs from all actions. 
Instead, we construct an SR-like statistic for each arm separately, and then show that the difference between the time-step and the sum of all SR-like statistics forms a martingale. Then, the rest of the steps follow the same line in Lemma 1 in \cite{veeravalli2024quickest}. The formal proof is included in Appendix \ref{sec:thm3}.

Again, Theorem 1 in \cite{veeravalli2024quickest} certifies the asymptotic optimality of the PA-UCB-CuSum procedure. In other words, asymptotically speaking ($\gamma \to \infty$), the per-action and the accumulated UCB-based CuSUM procedures have similar performance. We study their differences empirically, and specifically in the case where the distributions are unknown with generalized LLRs, in the next section.

\section{Experimental Study}
\label{sec:experiments}
\begin{figure*}[ht]
    \centering
    \begin{minipage}{0.48\textwidth}
        \centering
        % Change width to 1.0 to remove the 20% internal padding
        \includegraphics[width=0.7\linewidth]{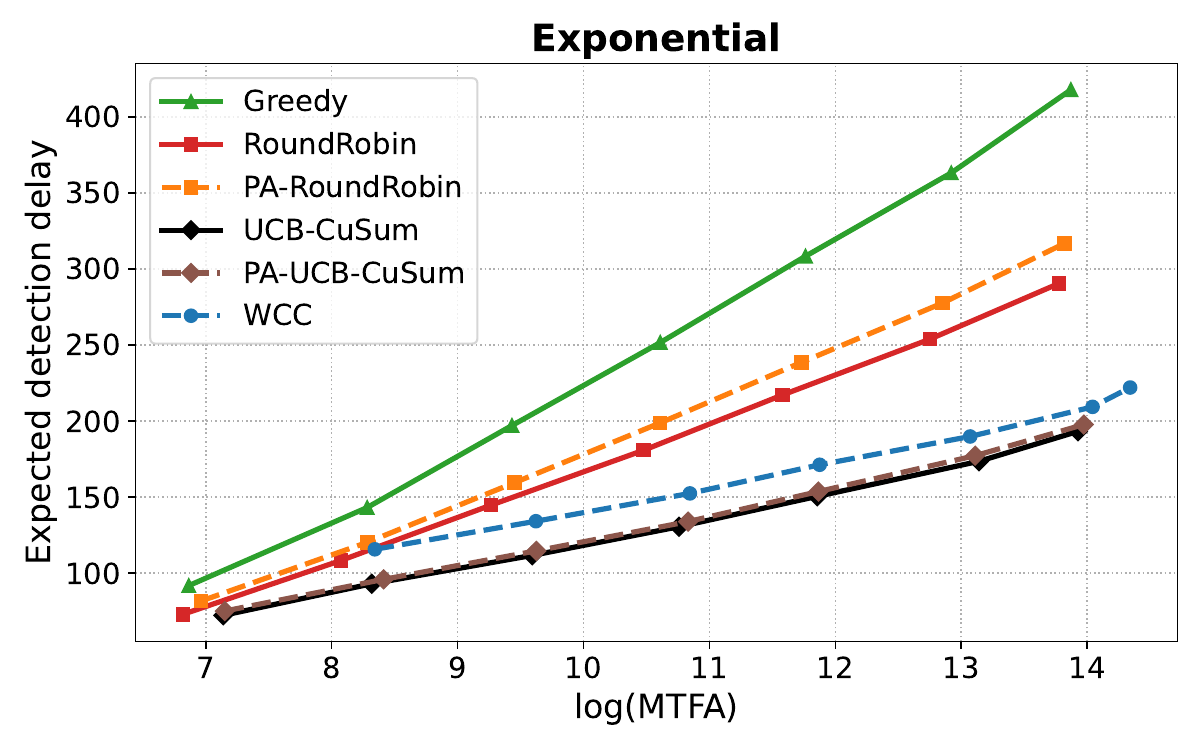}
    \end{minipage}
    \hspace{-1.8cm} % <--- Pulls the right column toward the left
    \begin{minipage}{0.48\textwidth}
        \centering
        \includegraphics[width=0.7\linewidth]{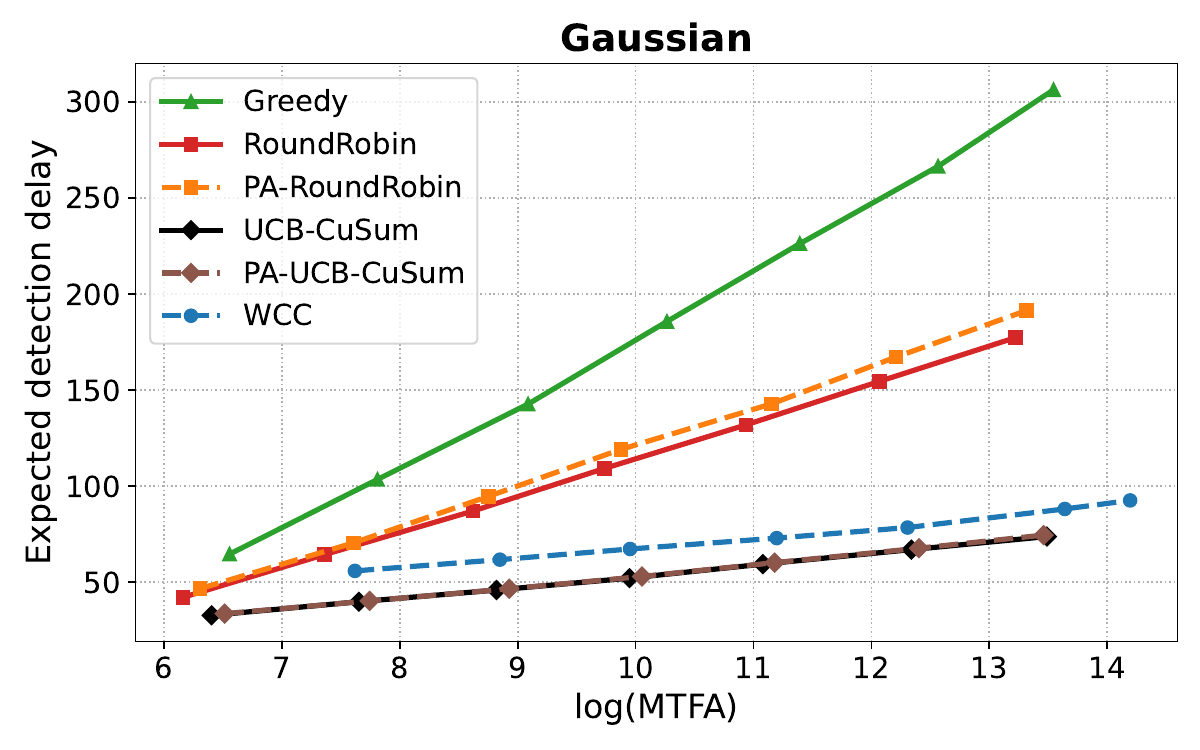}
    \end{minipage}

    %\vspace{-0.2cm} % Pulls the bottom row up slightly

    \begin{minipage}{0.48\textwidth}
        \centering
        \includegraphics[width=0.7\linewidth]{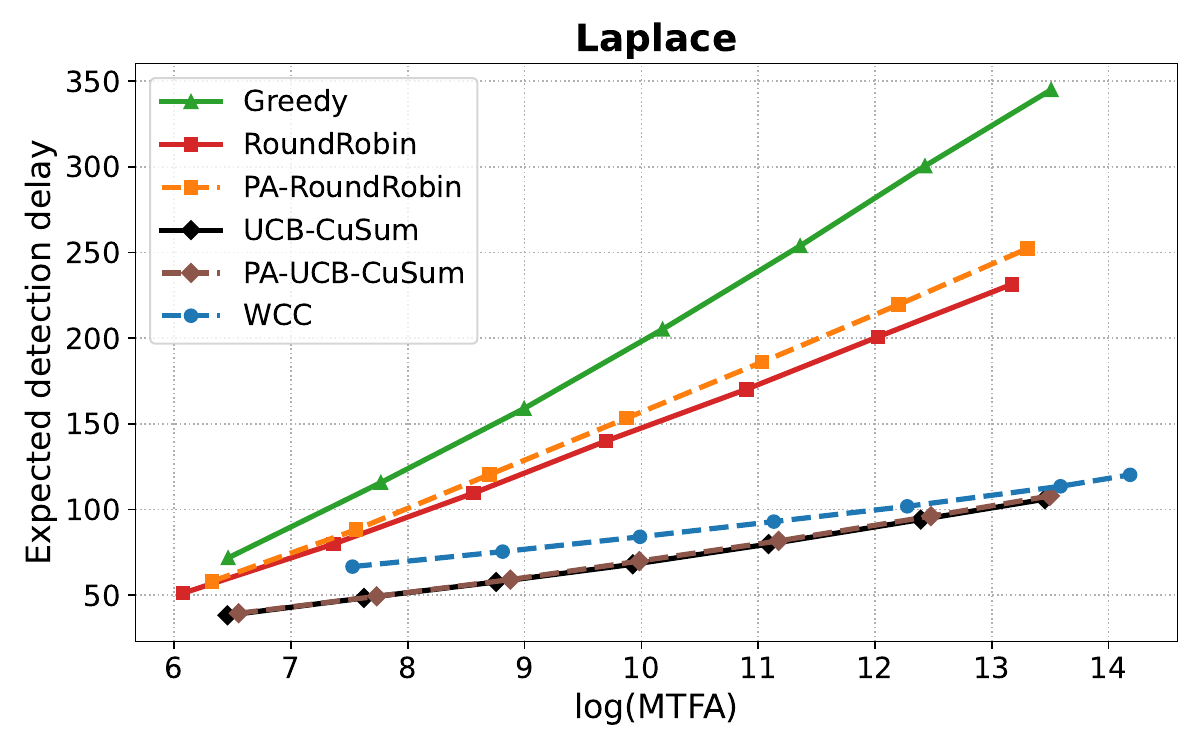}
    \end{minipage}
    \hspace{-1.8cm} % <--- Pulls the right column toward the left
    \begin{minipage}{0.48\textwidth}
        \centering
        \includegraphics[width=0.7\linewidth]{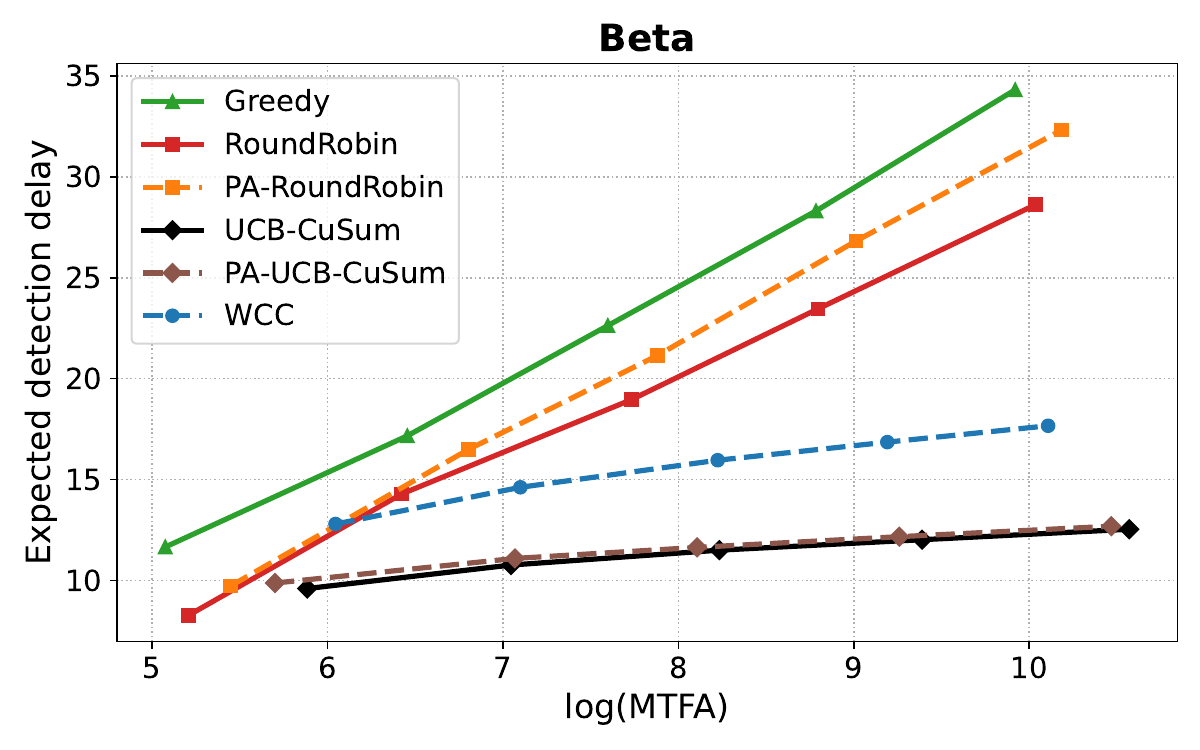}
    \end{minipage}

    \caption{Expected detection delay versus $\log(\mathrm{MTFA})$ (lower=better).
    Top: Exponential (left), Gaussian (right).
    Bottom: Laplace (left), LogNormal (right).}
    \label{fig:delay-mtfa}
\end{figure*}

In this section, we numerically evaluate how UCB-CuSum and the PA-UCB-CuSum algorithms on the multichannel bandit QCD problem using synthetic data. We further compare their performance against the leading approaches from prior works. 
% perform an experimental study for the studied . The experiments aim to highlight the efficacy of our two procedures under different kinds of channel densities.

% both in terms of performance as well as theoretical optimality.
\subsection{Algorithmic Benchmarks}
We compare against four baselines.
First, we include \emph{WCC}~\cite{veeravalli2024quickest}, the state-of-the-art procedure that is first-order optimal and empirically strong.
Second, we consider the \emph{Greedy} procedure (as in~\cite{xu2023asymptotic}), which samples a single component while maintaining a CuSum statistic: it repeatedly pulls the current component until its cumulative LLR either crosses the detection threshold (in which case it declares an alarm) or resets to zero (in which case it discards past samples and switches to the next component).
This method is known to enjoy asymptotic optimality guarantees in special cases, e.g., when only one component is affected or when all affected components share the same post-change distribution \cite{xu2021optimum,chaudhuri2021sequential}.
Finally, as an additional benchmark, we include \emph{RoundRobin}. This algorithm cyclically samples all components and aggregates the resulting LLRs into a single global statistic, declaring an alarm once this statistic exceeds its threshold. To make a comparison with the per-action variant of UCB-CuSum, we also include a per-action RoundRobin (PA-RoundRobin) which maintains a unique statistic for each action. Greedy and RoundRobin algorithms offer no optimality guarantees.

\subsection{Experimental Setup} 
We consider a bandit QCD problem with $K=10$ independent action-channels, where selecting action $a\in[K]$ at time $n$ yields an observation $X_{a,n}$. We consider four different observation distributions:  (i) Gaussian, (ii) Exponential, (iii) Laplace, and (iv) Beta, where only the last distribution family yields bounded observations in $[0,1]$.
% , enabling controlled dynamic-range rewards and a complementary regime to the unbounded models above.
In the pre-change regime ($n<\nu$), the observations are i.i.d. across time with
$X_{a,n}\sim\mathcal{N}(0,1)$ in the Gaussian case,
$X_{a,n}\sim\mathrm{Exp}(1)$ (mean $1$) in the Exponential case,
$X_{a,n}\sim\mathrm{Laplace}(0,1)$ in the Laplace case,
and for Beta, $X_{a,n}\sim\mathrm{Beta}(\alpha,\beta)$ with and mean equal to $0.01$, where throughout we set $\alpha+\beta=2$.
At the change-point $\nu$, for all cases except Beta (since it is bounded in $[0,1]$) a sparse subset of the channels becomes anomalous according to the shift vector
\begin{equation*}
\xi = [0,\,0,\,0.1,\,0,\,0,\,0.1,\,0,\,0,\,1,\,0]^\top .
\end{equation*}
and the shift is added to the corresponding actions. With the Beta distributed observations, we replace $0.1$ with $0.04$ and $1$ with $0.19$.
This design is intentionally \emph{sparse} and \emph{heterogeneous}: multiple components change, yet the magnitudes are unequal (two mild changes at level $0.1$ and one strong change at level $1$).
As a result, greedy selection rules can be brittle: once the Greedy policy locks onto a mildly-changed channel, it may continue sampling it and accumulate evidence too slowly, thereby delaying discovery of the strongly-changed channel.
This behavior is consistent with the fact that Greedy does not generally admit asymptotic optimality guarantees in such heterogeneous sparse settings, even though it can perform well empirically in simpler regimes \cite{veeravalli2024quickest}.
Notice that RoundRobin continues to probe all channels, preventing permanent lock-in; in our setting, the presence of a sufficiently strong change ($\xi_9=1$) suggests that periodic exploration can yield informative post-change samples and enables better detection compared to Greedy.

\begin{figure*}[ht]
    \centering
    \includegraphics[width=0.55\linewidth]{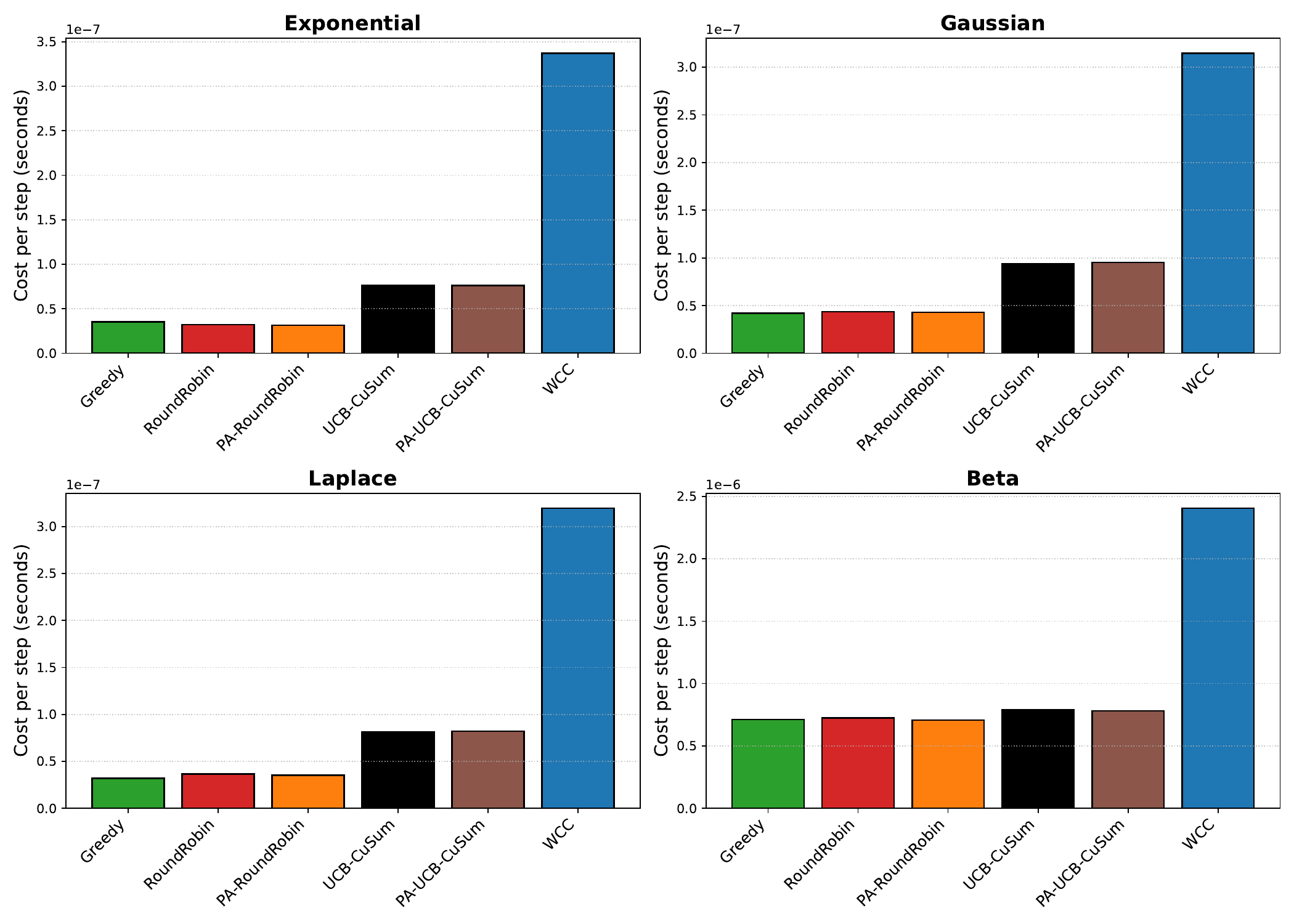}
    \caption{Average computational cost per iteration (seconds, lower=better). Top: Exponential (left), Gaussian (right).
    Bottom: Laplace (left), LogNormal (right).}
    \label{fig:comp-cost}

\end{figure*}

\subsection{Experimental Parameters and Metrics} 
We follow the parameter choices recommended in the original works, setting $w=\lceil 5\log b \rceil$ and $q=\lceil \log w\rceil$ for WCC. As per our theory, we set $W=\lceil 8\log b \rceil$ for Bandit-CuSum and PA-Bandit-CuSum. 
Greedy and RoundRobin do not have associated parameter choices.

Performance is evaluated via Monte-Carlo simulation in terms of (i) mean time to false alarm (MTFA) and (ii) expected detection delay, averaging over $200$K independent trials.
To generate MTFA--delay curves, we sweep the detection threshold to obtain different operating points; in practice, this can be done by selecting target MTFA levels $\gamma$ and setting $b=\log\gamma$.
In all figures, we report expected detection delay versus $\log(\mathrm{MTFA})$. In addition to statistical performance, we also report computational efficiency, quantified as the average \emph{cost per step} (wall-clock time per algorithm iteration), averaged over simulation runs.

\subsection{Experimental Results} 
Figure~\ref{fig:delay-mtfa} reports the trade-off between expected detection delay and $\log(\mathrm{MTFA})$. Across the different observation models, UCB-CuSum and PA-UCB-CuSum attain the best performance, achieving the smallest delay for a given MTFA level, consistently improving over WCC and substantially outperforming the Greedy and RoundRobin baselines. This behavior is consistent with heterogeneous sparse-change design: Greedy can lock onto mildly-changed channels and thus accumulates evidence slowly, while RoundRobin allocates a fixed fraction of samples to unaffected channels (which improves over Greedy), leading to larger delay. 

In contrast, UCB-CuSum adaptively concentrates sampling on informative channels with the largest changes, while retaining sufficient exploration to avoid lock-in. UCB-CuSum slightly outperforms its per-action variant PA-UCB-CuSum, as expected since it aggregates all observations into a single detection statistic. Importantly however, the performance gap between UCB-CuSum and PA-UCB-CuSum is small, indicating strong robustness to per-action decomposition. In contrast, PA-RoundRobin exhibits a substantially larger degradation relative to RoundRobin.

Figure~\ref{fig:comp-cost} compares computational cost per iteration (seconds per time-step).
All methods are highly efficient in absolute terms (on the order of $10^{-6}-10^{-7}$ seconds per step in our implementation). As expected, Greedy and RoundRobin are the cheapest per iteration, while WCC incurs the highest per-step cost; UCB-CuSum and PA-UCB-CuSum lie in between, offering a favorable accuracy-efficiency trade-off in practice. Notably, UCB-CuSum and PA-UCB-CuSum are markedly more efficient than WCC, achieving substantially lower per-step runtime while also attaining better statistical performance. 

\section{Towards Unknown Distributions} \label{secHo:unknown}
\label{sec:unknown}
The procedures above rely on the knowledge of the distributions  $(f_{a,0},f_{a,1})$ to compute  $\mathrm{LLR}(a,x)$ and the CuSum statistics. In a variety of applications, such distributions are not known. 
As a first step toward that scenario, we replace $\mathrm{LLR}$ by a per-action generalized likelihood ratio (GLR) statistic and retain the same
UCB-driven controlled sensing template for QCD. Specifically, we use the Bernoulli GLR statistic of \cite{besson2022efficient}:
given observations $X_1,\dots,X_n\in[0,1]$, define $\hat\mu_{t_1:t_2}:=\frac{1}{t_2-t_1+1}\sum_{t=t_1}^{t_2}X_t$ and $
\mathrm{kl}(p,q):=p\ln\!\frac{p}{q}+(1-p)\ln\!\frac{1-p}{1-q}.$
Then the GLR statistic $\mathrm{GLR}(n)$ is
\[
\max_{1\le s< n}\{\, s\,\mathrm{kl}(\hat\mu_{1:s},\hat\mu_{1:n})
+(n-s)\,\mathrm{kl}(\hat\mu_{s+1:n},\hat\mu_{1:n})\}.
\]
Based on the statistic, it is evident that observations need to share the same distribution in the segments used for the empirical mean. Therefore, the WCC and UCB-CuSum procedures cannot readily exploit the GLR statistic, as the statistic used for change detection combines likelihood ratios from different data streams, each of which can have different distributions, and hence, different GLRs. To address this challenge, we consider a  \emph{per-action} variant, where we maintain a different GLR statistic for each action, for which observations follow a  specific pre- and post-change distribution pair. We evaluate its performance empirically.

We modify the PA-UCB-CuSum procedure and utilize the Bernoulli GLR statistic in place of the LLR. We call it the PA-UCB-GLR procedure. The reward of each arm in this case is defined as the GLR statistic using the samples of that arm, normalized by the number of samples belonging to that specific arm. In order to estimate a measure of variance, we compute the empirical variance of the differences of the GLR statistics at each time step. Concretely, let $L_a$ be the total number of pulls of arm $a$ (not reset), and let $G_a(m)$ denote the GLR statistic computed from the first $m$ samples of arm $a$.
We use the normalized reward
\[
R_a(m):=\frac{G_a(m)}{m}.
\]
To quantify variability, define increments $\Delta G_a(m):=G_a(m)-G_a(m-1)$ for $m\ge2$ and the empirical variance
\begin{align*}
\widehat{\sigma}^2_a(L_a)
&:= \frac{1}{L_a-2}\sum_{m=2}^{L_a}\Big(\Delta G_a(m)-\overline{\Delta G}_a(L_a)\Big)^2,
\\
\overline{\Delta G}_a(L_a)
&:=\frac{1}{L_a-1}\sum_{m=2}^{L_a}\Delta G_a(m).
\end{align*}

\begin{algorithm}[htb]
    \small
  \caption{PA-UCB-GLR Procedure}
  \label{alg:PA_UCB_B_GLR}
  \begin{algorithmic}
    \STATE {\bfseries Input:} threshold $b$, interval length $W$.
    \STATE Initialize $n\leftarrow 0$. For all $a\in[K]$: $L_a\leftarrow 0$, $G_a\leftarrow 0$.
    \STATE For all $a\in[K]$: $N_a\leftarrow 0$, $\widehat\mu_a\leftarrow 0$, $\widehat\sigma_a^2\leftarrow 0$, $\mathrm{UCB}_a\leftarrow +\infty$.
    \REPEAT
      \IF{$n \bmod W = 0$}
        \STATE Reset $(N_a,\widehat\mu_a,\widehat\sigma_a^2,\mathrm{UCB}_a)$ for all $a\in[K]$
      \ENDIF
      \STATE Choose action $A \in \arg\max_{a\in[K]} \mathrm{UCB}_a$.
      \STATE Receive observation $X$ from channel $A$.
      \STATE $L_A \leftarrow L_A + 1$.
      \STATE Update $G_A \leftarrow \mathrm{GLR}_A(L_A)$
      \STATE Set reward $R \leftarrow G_A/L_A$.
      \STATE Update $\widehat\sigma_A^2 \leftarrow \widehat\sigma_A^2(L_A)$
      \STATE Update $N_A,\widehat\mu_A$ using reward $R$.
      \STATE Update $\mathrm{UCB}_A \leftarrow \widehat\mu_A + \sqrt{(2\,\widehat\sigma_A^2 \log W)/N_A}$.
      \STATE $n \leftarrow n + 1$.
    \UNTIL{$G_a \ge b$ for some $a\in[K]$}
  \end{algorithmic}
\end{algorithm}
% \begin{figure}[ht]
%     \centering
%     \includegraphics[width=0.5\linewidth]{CDC2026/figures/qcd_glr.pdf}
%     \caption{Expected detection delay versus $\log(\mathrm{MTFA})$ (lower=better) for PA-RoundRoubin and PA-UCB using the GLR detector with Bernoulli observations.}
%     \label{fig:delay-mtfa-glr-bernoulli}
% \end{figure}

\begin{figure}[ht]
    \centering
    \includegraphics[width=0.9\linewidth]{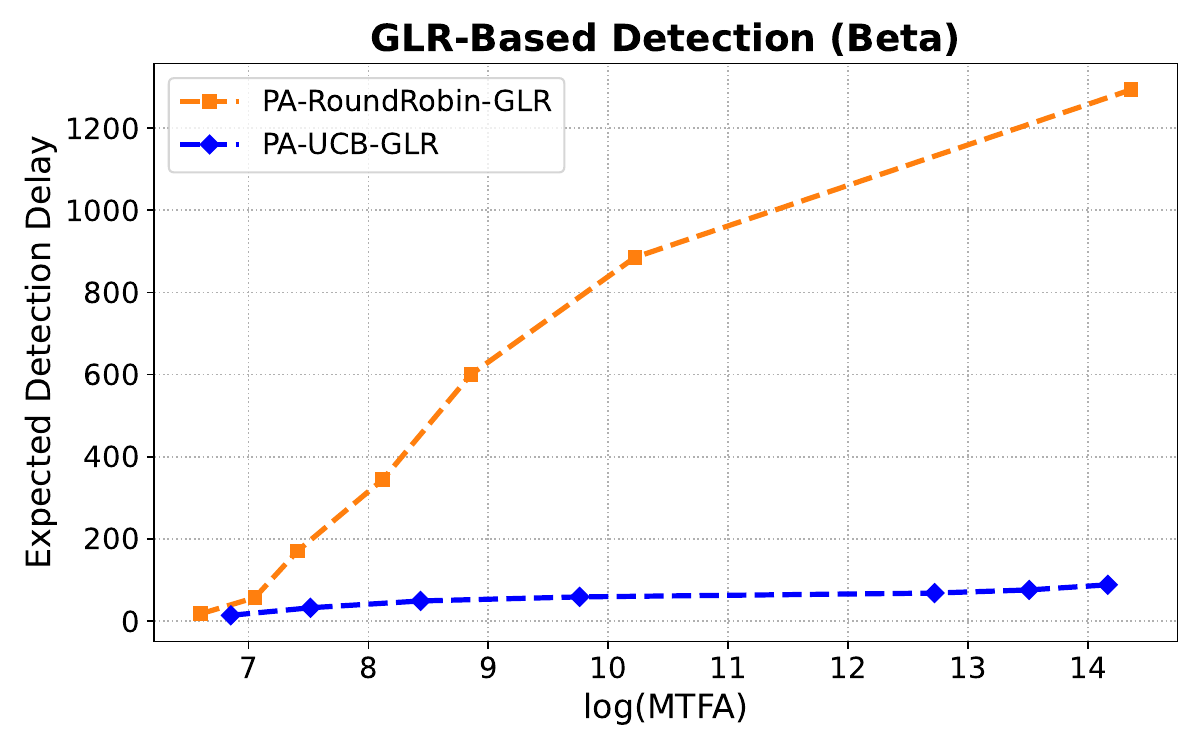}
    \caption{Expected detection delay versus $\log(\mathrm{MTFA})$ (lower=better) for PA-RoundRoubin and PA-UCB using the GLR detector with Beta observations.}
    \label{fig:delay-mtfa-glr-beta}
\end{figure}

To evaluate our the PA-UCB-GLR procedure, we employ the experiments with Beta distributions used in Section \ref{sec:experiments} and compare with a GLR variant of the PA-RoundRobin.
Since observations are bounded in $[0,1]$, we instantiate the per-arm statistic $G_a(\cdot)$ with the Bernoulli GLR, computed {separately} on each arm’s sample stream.
Both methods use the same stopping rule shown in the end of Algorithm (threshold $b$), and differ only in how samples are allocated across arms:
\emph{PA-RoundRobin-GLR} cycles through arms uniformly, whereas \emph{PA-UCB-GLR} (Alg.~\ref{alg:PA_UCB_B_GLR}) adaptively selects arms via a UCB index built from the normalized reward and the empirical variance of GLR increments.
Figure~\ref{fig:delay-mtfa-glr-beta} shows that PA-UCB-GLR achieves a significantly better MTFA--delay trade-off, with an increasingly larger advantage at higher MTFA: round-robin spends a fixed fraction of samples on weakly-informative arms, while PA-UCB-GLR concentrates sensing on the most informative arm, substantially reducing delay.

\section{Conclusions}
In this paper, we developed {simple and efficient} multichannel bandit QCD procedures that (i) achieve optimal MTFA--delay scaling in the known-model setting and (ii) extend naturally to \emph{unknown} pre/post distributions.
% ---motivated by detection-augmented learning in piecewise-stationary environments, where exploration actions act as ``channels'' and change detection triggers restarts.
Specifically, when the distributions of observations from each action is known, we proposed and analyzed \emph{UCB-CuSum} and \emph{PA-UCB-CuSum}, which treat per-sample LLRs as bandit rewards and used periodically-restarted UCB to select actions. Such a policy focuses  on more informative affected channel(s) quickly, while enforcing enough exploration to avoid lock-in to pre-change optimal actions. These procedures were shown to achieve first-order asymptotic optimality.
% for both procedures, matching the fundamental lower bound.
Across multiple observation models with sparse, heterogeneous changes, the our methods attained the best MTFA--delay trade-off, while remaining  computationally efficient relative to state-of-art. 
% UCB-CuSum gains slightly from aggregating evidence, but the small gap to \emph{PA-UCB-CuSum} shows strong robustness of the per-action decomposition
% ---useful in piecewise-stationary settings that favor modular, action-wise detectors. 
Finally, we introduced \emph{PA-UCB-GLR} that replaces per-action LLR with a Bernoulli GLR statistic, extending our algorithm design philosophy to distribution-agnostic setups. 
% ; since GLR requires homogeneous samples within each arm, this underscores the advantage of per-action designs for distribution-agnostic extensions. 
While we relegate its theoretical analysis to future work, our experiments confirmed its efficacy over round-robin action selection policy.

\bibliography{ref}
\bibliographystyle{IEEEtran}

%%%%%%%%%%%%%%%%%%%%%%%%%%%%%%%%%%%%%%%%%%%%%%%%%%%%%%%%%%%%%%%%%%%%%%%%%%%%%%%
%%%%%%%%%%%%%%%%%%%%%%%%%%%%%%%%%%%%%%%%%%%%%%%%%%%%%%%%%%%%%%%%%%%%%%%%%%%%%%%
% APPENDIX
%%%%%%%%%%%%%%%%%%%%%%%%%%%%%%%%%%%%%%%%%%%%%%%%%%%%%%%%%%%%%%%%%%%%%%%%%%%%%%%
%%%%%%%%%%%%%%%%%%%%%%%%%%%%%%%%%%%%%%%%%%%%%%%%%%%%%%%%%%%%%%%%%%%%%%%%%%%%%%%
\newpage
\onecolumn
\appendix
\section{Proof of Theorem \ref{thm:UCB_CuSum_opt}} \label{sec:thm1}

To prove that $\E_{\infty}^{A_{\mrm{UCB}}} \lb T_{\log \gamma}^{A_{\mrm{UCB}}} \rb \geq \gamma$, we can apply Lemma 1 in \cite{veeravalli2024quickest} and directly obtain the lower bound.

To prove \eqref{eq:thm_WADD}, we show that as $W, b \to \infty$ with $W = o(b)$,
\begin{align}\label{eq:UCB_CuSum_WADD_upp}
\mcal{J}_{\mcal{A}} \lp A_{\mrm{UCB}}, T_{b}^{A_{\mrm{UCB}}} \rp \leq \frac{b}{I_{\mcal{A}}}\lp 1 + o \lp 1 \rp \rp.
\end{align}
The proof of this upper bound on $\mcal{J}_{\mcal{A}}$ requires the following lemma.

\begin{lemma}\label{lem:nu_to_one}
For any change-point $\nu \in \mbb{N}$ and $\mcal{A} \subseteq \lb K \rb$,
\begin{equation}\label{eq:lem1}
    \E_{\nu,\mcal{A}}^{A_{\mrm{UCB}}} \lb \lp T_{b}^{A_{\mrm{UCB}}} - \nu + 1\rp^{+} \Big| \mcal{F}_{\nu-1} \rb \leq W + \E_{1,\mcal{A}}^{A_{\mrm{UCB}}} \lb T_{b}^{A_{\mrm{UCB}}} \rb.
\end{equation}
\end{lemma}

\begin{proof}
Let $\ell \coloneqq \argmin \lbp i \in \mbb{N} : i \geq \nu \mrm{\:and}\enspace i \mod W = 1 \rbp$ be the first time-step at which the UCB algorithm restarts after the change occurs. Then, for $n \geq \ell$, define the statistic with the recursion
\begin{equation}\label{eq:windowed_CuSum}
    C_{\ell:n}^{A_{\mrm{UCB}}} \coloneqq \max \lbp C_{\ell:n - 1}^{A_{\mrm{UCB}}}, 0 \rbp + \mrm{LLR} \lp A_{n}, X_{A_{n},n} \rp
\end{equation}
with $C_{\ell: \ell - 1}^{A_{\mrm{UCB}}} \coloneqq 0$. The statistic in \eqref{eq:windowed_CuSum} is the CuSum-like statistic that accumulates the LLRs of samples after time-step $\ell$. We also define a new stopping time that uses this statistic in place of the CuSum-like statistic in \eqref{eq:CuSum_stat}:
\begin{equation}\label{eq:windowed_T}
    T_{b, \ell}^{A_{\mrm{UCB}}} \coloneqq \inf \lbp n \geq \ell: C_{\ell:n}^{A_{\mrm{UCB}}} \geq b \rbp.
\end{equation}
Then,
\begin{align}
    \E_{\nu,\mcal{A}}^{A_{\mrm{UCB}}} \lb \lp T_{b}^{A_{\mrm{UCB}}} - \nu + 1\rp^{+} \Big| \mcal{F}_{\nu-1} \rb &\overset{(a)}{\leq} \E_{\nu,\mcal{A}}^{A_{\mrm{UCB}}} \lb \lp T_{b, \ell}^{A_{\mrm{UCB}}} - \nu + 1\rp^{+} \Big| \mcal{F}_{\nu-1} \rb \nonumber\\
    &= \E_{\nu,\mcal{A}}^{A_{\mrm{UCB}}} \lb T_{b, \ell}^{A_{\mrm{UCB}}} - \nu + 1 \Big| \mcal{F}_{\nu-1} \rb \nonumber\\
    &= \ell - \nu + \E_{\nu,\mcal{A}}^{A_{\mrm{UCB}}} \lb T_{b, \ell}^{A_{\mrm{UCB}}} - \ell + 1 \Big| \mcal{F}_{\nu-1} \rb \nonumber\\
    &\overset{(b)}{\leq} W + \E_{\nu,\mcal{A}}^{A_{\mrm{UCB}}} \lb T_{b, \ell}^{A_{\mrm{UCB}}} - \ell + 1 \Big| \mcal{F}_{\nu-1} \rb \nonumber\\
    &\overset{(c)}{=} W + \E_{\nu,\mcal{A}}^{A_{\mrm{UCB}}} \lb T_{b, \ell}^{A_{\mrm{UCB}}} - \ell + 1 \rb \nonumber\\
    &\overset{(d)}{=} W + \E_{\ell,\mcal{A}}^{A_{\mrm{UCB}}} \lb T_{b, \ell}^{A_{\mrm{UCB}}} - \ell + 1 \rb \nonumber\\
    &\overset{(e)}{=} W + \E_{1,\mcal{A}}^{A_{\mrm{UCB}}} \lb T_{b, 1}^{A_{\mrm{UCB}}} \rb \nonumber\\
    &\overset{(f)}{=} W + \E_{1,\mcal{A}}^{A_{\mrm{UCB}}} \lb T_{b}^{A_{\mrm{UCB}}} \rb. \label{eq:CuSum_Lem1_laststep}
\end{align}
In step $(a)$, $T_{b, \ell}^{A_{\mrm{UCB}}} \geq T_{b}^{A_{\mrm{UCB}}}$ almost surely, as we can show that $C_{\ell:n}^{A_{\mrm{UCB}}} \leq C_{n}^{A_{\mrm{UCB}}}$ for $n \geq \ell$ by induction. In step $(b)$, we leverage the fact that $ \ell - W \leq \nu$, as $\ell$ is the first time-step at which the UCB algorithm restarts after the change. In step $(c)$, because the statistic $C_{\ell:n}^{A_{\mrm{UCB}}}$ is independent of the samples prior to the time-step $\ell - 1$, we can utilize the fact that $T_{b, \ell}^{A_{\mrm{UCB}}}$ is independent of $\mcal{F}_{\ell - 1}$ and that $\mcal{F}_{\nu - 1} \subseteq \mcal{F}_{\ell - 1}$. Step $(d)$ follows from the fact that $X_{a,n}$ follows the post-change density $f_{a,1}$ for $n \geq \ell$ and $a \in \mcal{A}$ when the change occurs either at $\nu$ or at $\ell$. Step $(e)$ results from the fact that $X_{a,n}$ follows the post-change density $f_{a,1}$ for $n \geq \ell$ and $a \in \mcal{A}$ under $\mbb{P}_{\ell, \mcal{A}}^{A_{\mrm{UCB}}}$, and that $X_{a,n}$ follows the post-change density $f_{a,1}$ for $n \geq 1$ and $a \in \mcal{A}$ under $\mbb{P}_{1, \mcal{A}}^{A_{\mrm{UCB}}}$. In step $(f)$, $T_{b}^{A_{\mrm{UCB}}}=T_{b, 1}^{A_{\mrm{UCB}}}$ by the definitions in \eqref{eq:CuSum_stat} and \eqref{eq:windowed_CuSum}.
\end{proof}

With Lemma \ref{lem:nu_to_one}, we can derive that
\begin{align}
    \mcal{J}_{\mcal{A}} \lp A_{\mrm{UCB}}, T_{b}^{A_{\mrm{UCB}}} \rp \leq W + \E_{1,\mcal{A}}^{A_{\mrm{UCB}}} \lb T_{b}^{A_{\mrm{UCB}}} \rb.
\end{align}
To upper bound $\E_{1,\mcal{A}}^{A_{\mrm{UCB}}} \lb T_{b}^{A_{\mrm{UCB}}} \rb$, we need the following lemma.

\begin{lemma}\label{lem:E1_upp_UCB_CuSum}
For some constant $c > 0$ that does not depend on $b, I_{\mcal{A}}$, and $W$,
\begin{equation}
    \E_{1,\mcal{A}}^{A_{\mrm{UCB}}} \lb T_{b}^{A_{\mrm{UCB}}} \rb \leq \frac{W \lp b + c \lp 1 + \sqrt{b} \rp \rp}{WI_{\mcal{A}} -  \sum_{a: \Delta_{a} > 0} \lp 3\Delta_{a} + \frac{16 \log W}{\Delta_{a}} \rp}
\end{equation}
where $\Delta_{a} = I_{\mcal{A}} - D \lp f_{a,1} || f_{a,0} \rp$.
\end{lemma}

\begin{proof}
For any $n \in \mbb{N}$, define the statistic $\Bar{C}_{n}^{A_{\mrm{UCB}}}$ to be the sum of the LLRs up to time-step $n$, i.e.,
\begin{align}\label{eq:Sum_LLR}
    \Bar{C}_{n}^{A_{\mrm{UCB}}}\coloneqq \sum_{i=1}^{n} \mrm{LLR}\lp A_{i}, X_{A_{i}, i}\rp.
\end{align}
Also, we define $\bar{T}_{b}^{A_{\mrm{UCB}}}$ to be the stopping time at which $\Bar{C}_{n}^{A_{\mrm{UCB}}}$ crosses the threshold $b$ at the end of an $W$-length interval $\mcal{I}_{i} = \lbp (i-1)W + 1, \dots, iW \rbp$, i.e.,
\begin{align}\label{eq:Sum_LLR_stop}
    \bar{L}_{b}^{A_{\mrm{UCB}}} \coloneqq \inf \lbp i \in \mbb{N}: \Bar{C}_{iW}^{A_{\mrm{UCB}}} \geq b \rbp \mrm{\; and \;} \Bar{T}_{b}^{A_{\mrm{UCB}}} \coloneqq \bar{L}_{b}^{A_{\mrm{UCB}}}W.
\end{align}
Since $\bar{C}_{n}^{A_{\mrm{UCB}}} \leq C_{n}^{A_{\mrm{UCB}}}$ for all $n \in \mbb{N}$, $\Bar{T}_{b}^{A_{\mrm{UCB}}} \geq T_{b}^{A_{\mrm{UCB}}}$. Now, for any $i \in \mbb{N}$, define the filtration $\mcal{G}_{i} \coloneqq \sigma \lp X_{1}, \dots, X_{iW} \rp$ and the following statistic to be the difference of $\bar{C}_{n}^{A_{\mrm{UCB}}}$ over interval $\mcal{I}_{i} = \lbp \lp i-1\rp W + 1, \dots, iW \rbp$, i.e.,
\begin{equation}
    Y_{i}^{A_{\mrm{UCB}}} \coloneqq \Bar{C}_{iW}^{A_{\mrm{UCB}}} - \Bar{C}_{\lp i - 1 \rp W}^{A_{\mrm{UCB}}} = \sum_{j=(i-1)W + 1}^{iW} \mrm{LLR}\lp A_{j}, X_{A_{j}, j}\rp.
\end{equation}
Recall that $I_{\mcal{A}} \coloneqq \max_{a \in \mcal{A}} D \lp f_{a, 1} || f_{a,0} \rp$. Then,
\begin{equation}
    \E_{1,\mcal{A}}^{A_{\mrm{UCB}}} \lb Y_{i}^{A_{\mrm{UCB}}} \Big| \mcal{G}_{i-1} \rb \overset{(a)}{=} \E_{1,\mcal{A}}^{A_{\mrm{UCB}}} \lb Y_{i}^{A_{\mrm{UCB}}} \rb = \E_{1,\mcal{A}}^{A_{\mrm{UCB}}} \lb \sum_{j=(i-1)W + 1}^{iW} \mrm{LLR}\lp A_{j}, X_{A_{j}, j}\rp \rb \overset{(b)}{\leq} WI_{\mcal{A}},
\end{equation}
and
\begin{equation}
    \E_{1,\mcal{A}}^{A_{\mrm{UCB}}} \lb Y_{i}^{A_{\mrm{UCB}}} | \mcal{G}_{i-1} \rb = \E_{1,\mcal{A}}^{A_{\mrm{UCB}}} \lb \sum_{j=(i-1)W + 1}^{iW} \mrm{LLR}\lp A_{j}, X_{A_{j}, j}\rp \rb \overset{(c)}{\geq} WI_{\mcal{A}} - \sum_{a: \Delta_{a} > 0} \lp 3\Delta_{a} +  \frac{16 \log W}{\Delta_{a}} \rp.
\end{equation}
In step $(a)$, we leverage the fact that $Y_{i}^{A_{\mrm{UCB}}}$ is independent of $\mcal{G}_{i-1}$. In step $(b)$, we apply the fact that $\E_{1,\mcal{A}}^{A_{\mrm{UCB}}} \lb \mrm{LLR}\lp A_{j}, X_{A_{j}, j}\rp \rb \leq I_{\mcal{A}}$. In step $(c)$, we apply the regret upper bound in Theorem 7.1 in \cite{lattimore2020bandit}, as the LLRs are the rewards of the UCB algorithm. Recall that we assume $\int \lp \log \lp f_{a,1} / f_{a,0} \rp - D \lp f_{a,1} || f_{a,0} \rp \rp^{2} f_{a,1}d\lambda = v < \infty$ for any $a \in \lb K \rb$. Then,
\begin{align}
    &\E_{1,\mcal{A}}^{A_{\mrm{UCB}}} \lb \lp Y_{i}^{A_{\mrm{UCB}}} \rp^{2} \bigg| \mcal{G}_{i-1} \rb \nonumber\\
    &\overset{(a)}{=} \E_{1,\mcal{A}}^{A_{\mrm{UCB}}} \lb \lp Y_{i}^{A_{\mrm{UCB}}} \rp^{2} \rb \nonumber\\
    &= \E_{1,\mcal{A}}^{A_{\mrm{UCB}}} \lb \sum_{j=(i-1)W + 1}^{iW} \sum_{k=(i-1)W + 1}^{iW} \mrm{LLR}\lp A_{j}, X_{A_{j}, j}\rp \mrm{LLR}\lp A_{k}, X_{A_{k}, k}\rp \rb \nonumber\\
    &= \sum_{j=(i-1)W + 1}^{iW} \sum_{k=(i-1)W + 1}^{iW} \E_{1,\mcal{A}}^{A_{\mrm{UCB}}} \lb \mrm{LLR}\lp A_{j}, X_{A_{j}, j}\rp \mrm{LLR}\lp A_{k}, X_{A_{k}, k}\rp \rb \nonumber\\
    &= \sum_{j=(i-1)W + 1}^{iW} \E_{1,\mcal{A}}^{A_{\mrm{UCB}}} \lb \lp\mrm{LLR}\lp A_{j}, X_{A_{j}, j}\rp\rp^{2} \rb + 2\sum_{j=(i-1)W + 1}^{iW} \sum_{k=j+1}^{iW} \E_{1,\mcal{A}}^{A_{\mrm{UCB}}} \lb \mrm{LLR}\lp A_{j}, X_{A_{j}, j}\rp \mrm{LLR}\lp A_{k}, X_{A_{k}, k}\rp \rb \nonumber \\
    &\overset{(b)}{\leq} Wv + W^{2}I_{\mcal{A}}^{2}
\end{align}
where step $(a)$ stems from the fact that $Y_{i}^{A_{\mrm{UCB}}}$ is independent of $\mcal{G}_{i-1}$, and step $(b)$ follows from Assumption \ref{assum:KL}. Then, we can apply Proposition 1 in \cite{veeravalli2024quickest} and obtain that 
\begin{equation}\label{eq:L_upp}
    \E_{1,\mcal{A}}^{A_{\mrm{UCB}}} \lb \bar{L}_{b}^{A_{\mrm{UCB}}} \rb \leq \frac{b + c \lp 1 + \sqrt{b} \rp}{WI_{\mcal{A}} -  \sum_{a: \Delta_{a} > 0} \lp 3\Delta_{a} + \frac{16 \log W}{\Delta_{a}} \rp}.
\end{equation}
Thus, 
\begin{equation}\label{eq:final_upp}
    \E_{1,\mcal{A}}^{A_{\mrm{UCB}}} \lb T_{b}^{A_{\mrm{UCB}}} \rb \overset{(a)}{\leq} \E_{1,\mcal{A}}^{A_{\mrm{UCB}}} \lb \Bar{T}_{b}^{A_{\mrm{UCB}}} \rb \overset{(b)}{\leq} \frac{W \lp b + c \lp 1 + \sqrt{b} \rp \rp}{WI_{\mcal{A}} - \sum_{a: \Delta_{a} > 0} \lp 3\Delta_{a} + \frac{16 \log W}{\Delta_{a}} \rp}
\end{equation}
where step $(a)$ follows from the fact that $\bar{T}_{b}^{A_{\mrm{UCB}}} \geq T_{b}^{A_{\mrm{UCB}}}$, and step $(b)$ results from \eqref{eq:L_upp}.
\end{proof}

With Lemma \ref{lem:E1_upp_UCB_CuSum}, we derive
\begin{equation}
    \mcal{J}_{\mcal{A}} \lp A_{\mrm{UCB}}, T_{b}^{A_{\mrm{UCB}}} \rp \leq W + \frac{W \lp b + c \lp 1 + \sqrt{b} \rp \rp}{WI_{\mcal{A}} - \sum_{a: \Delta_{a} > 0} \lp 3\Delta_{a} + \frac{16 \log W}{\Delta_{a}} \rp}.
\end{equation}
Then, as $W, b \to \infty$ with $W = o(b)$,
\begin{equation}
     \mcal{J}_{\mcal{A}} \lp A_{\mrm{UCB}}, T_{b}^{A_{\mrm{UCB}}} \rp \leq \frac{b}{I_{\mcal{A}}} \lp 1 + o \lp 1 \rp \rp.
\end{equation}
This completes the proof of Theorem \ref{thm:UCB_CuSum_opt}.
%%%%%%%%%%%%%%%%%%%%%%%%%%%%%%%%%%%%%%%%%%%%%%%%%%%%%%%%%%%%%%%%%%%%%%%%%%%%%%%
%%%%%%%%%%%%%%%%%%%%%%%%%%%%%%%%%%%%%%%%%%%%%%%%%%%%%%%%%%%%%%%%%%%%%%%%%%%%%%%

\section{Proof of Theorem \ref{thm:PA_UCB_CuSum_opt}}
\label{sec:thm3}

To prove that $\E_{\infty}^{A_{\mrm{UCB}}} \lb \tilde{T}_{\log \gamma}^{A_{\mrm{UCB}}} \rb \geq \gamma$, we first construct the following SR-like statistic \cite{shiryaev1961problem} for each action: for $a \in \lb K \rb$ and $n \in \mbb{N}$,
\begin{equation}\label{eq:SR}
    S_{a,n}^{A_{\mrm{UCB}}} \coloneqq \begin{dcases}
        \lp S_{a,n-1}^{A_{\mrm{UCB}}} + 1 \rp \exp\lp\mrm{LLR} \lp a, X_{a,n} \rp \rp,& A_{n} = a\\
        S_{a,n-1}^{A_{\mrm{UCB}}},& \mrm{otherwise}
    \end{dcases}
\end{equation}
with $S_{a,0} \coloneqq 0$. In addition, we define the following statistic as the sum of the SR-like statistic in \eqref{eq:SR}: for each $n \in \mbb{N}$,
\begin{align} \label{eq:SR_Sum}
    S_{n} \coloneqq \sum_{a \in \lb K \rb} S_{a, n}.
\end{align}
Thereupon, we can see that $\lbp S_{n}: n \in \mbb{N} \rbp$ is adapted to the filtration $\lbp \mcal{F}_{n}: n \in \mbb{N} \rbp$. We show that $\lbp S_{n} - n: n \in \mbb{N} \rbp$ is a martingale with respect to the filtration $\lbp \mcal{F}_{n}: n \in \mbb{N} \rbp$ under $\E_{\infty}^{A_{\mrm{UCB}}}$: for any $n \in \mbb{N}$,
\begin{align}
    \E_{\infty}^{A_{\mrm{UCB}}} \lb S_{n} - n | \mcal{F}_{n-1} \rb &= \E_{\infty}^{A_{\mrm{UCB}}} \lb S_{A_{n},n} + \sum_{a \neq A_{n}} S_{a, n} - n \Bigg| \mcal{F}_{n-1} \rb \nonumber\\
    &= \E_{\infty}^{A_{\mrm{UCB}}} \lb \lp S_{A_{n},n - 1} + 1 \rp \exp \lp \mrm{LLR} \lp A_{n}, X_{A_{n}, M_{A_{n}, n}} \rp \rp + \sum_{a \neq A_{n}} S_{a, n-1} - n\Bigg| \mcal{F}_{n-1} \rb \nonumber\\
    &= \lp S_{A_{n},n - 1} + 1 \rp \E_{\infty}^{A_{\mrm{UCB}}} \lb \exp \lp \mrm{LLR} \lp A_{n}, X_{A_{n}, M_{A_{n}, n}} \rp \rp \big| \mcal{F}_{n-1} \rb + \sum_{a \neq A_{n}} S_{a, n-1} - n \nonumber\\
    &= \lp S_{A_{n},n - 1} + 1 \rp \E_{\infty}^{A_{\mrm{UCB}}} \lb f_{A_{n},1} \lp X_{A_{n}, n} \rp / f_{A_{n},0} \lp X_{A_{n}, n} \rp \big| \mcal{F}_{n-1} \rb + \sum_{a \neq A_{n}} S_{a, n-1} - n \nonumber\\
    &= \sum_{a \in \lb K \rb} S_{a, n-1} - \lp n - 1 \rp \nonumber\\
    &= S_{n-1} - \lp n - 1 \rp.
\end{align}
We can see that $S_{n} \geq S_{a, n} \geq \exp\lp C_{a, n} \rp$ for any $a \in \lb K \rb$ and $n \in \mbb{N}$. Therefore, 
\begin{align}
    \E_{\infty}^{A_{\mrm{UCB}}} \lb \tilde{T}_{\log \gamma}^{A_{\mrm{UCB}}} \rb \overset{(a)}{=} \E_{\infty}^{A_{\mrm{UCB}}} \lb S_{\tilde{T}_{\log \gamma}^{A_{\mrm{UCB}}}} \rb \geq \E_{\infty}^{A_{\mrm{UCB}}} \lb \exp \lp C_{\tilde{T}_{\log \gamma}^{A_{\mrm{UCB}}}} \rp \rb \overset{(b)}{\geq} \gamma.
\end{align}
where step $(a)$ follows from Optional Sampling Theorem, and step $(b)$ stems from the definition of $\tilde{T}_{\log \gamma}^{A_{\mrm{UCB}}}$ in \eqref{eq:CuSum_T_sep}.

To prove \eqref{eq:thm_PA_WADD}, we show that as $W, b \to \infty$ with $W = o(b)$,
\begin{align}\label{eq:UCB_CuSum_WADD_upp}
\mcal{J}_{\mcal{A}} \lp A_{\mrm{UCB}}, \tilde{T}_{b}^{A_{\mrm{UCB}}} \rp \leq \frac{b}{I_{\mcal{A}}}\lp 1 + o \lp 1 \rp \rp.
\end{align}
The proof of this upper bound on $\mcal{J}_{\mcal{A}}$ requires the following lemma similar to Lemma \ref{lem:nu_to_one}.

\begin{lemma}\label{lem:nu_to_one_PA}
For any change-point $\nu \in \mbb{N}$ and $\mcal{A} \subseteq \lb K \rb$,
\begin{equation}\label{eq:lem1}
    \E_{\nu,\mcal{A}}^{A_{\mrm{UCB}}} \lb \lp \tilde{T}_{b}^{A_{\mrm{UCB}}} - \nu + 1\rp^{+} \Big| \mcal{F}_{\nu-1} \rb \leq W + \E_{1,\mcal{A}}^{A_{\mrm{UCB}}} \lb \tilde{T}_{b}^{A_{\mrm{UCB}}} \rb.
\end{equation}
\end{lemma}

\begin{proof}
Recall that $\ell \coloneqq \argmin \lbp i \in \mbb{N} : i \geq \nu \mrm{\:and}\enspace i \mod W = 1 \rbp$ is the first time-step at which the UCB algorithm restarts after the change occurs. We define the CuSum-like statistic that accumulates the LLRs starting from time-step $\ell$ using the following recursion: for $n \geq \ell$ and $a \in \lb K \rb$, 
\begin{equation}\label{eq:windowed_CuSum_sep}
    C_{a,\ell:n}^{A_{\mrm{UCB}}} \coloneqq \begin{dcases}
        \max \lbp C_{a,\ell:n - 1}^{A_{\mrm{UCB}}}, 0 \rbp + \mrm{LLR} \lp a, X_{a,n} \rp,& A_{n} = a\\
        C_{a,\ell:n - 1}^{A_{\mrm{UCB}}},& \mrm{otherwise}
    \end{dcases}
\end{equation}
with $C_{a, \ell: \ell - 1}^{A_{\mrm{UCB}}} \coloneqq 0$. We also define a new stopping time that stops when $C_{a,\ell:n}^{A_{\mrm{UCB}}}$ crosses the threshold $b$: 
\begin{equation}\label{eq:windowed_T_sep}
    \tilde{T}_{b, \ell}^{A_{\mrm{UCB}}} \coloneqq \inf \lbp n \geq \ell: C_{a, \ell:n}^{A_{\mrm{UCB}}} \geq b \mrm{\:for\:some\:} a \in \lb K \rb \rbp.
\end{equation}
Then, following similar steps in \eqref{eq:CuSum_Lem1_laststep},
\begin{align}
    \E_{\nu,\mcal{A}}^{A_{\mrm{UCB}}} \lb \lp \tilde{T}_{b}^{A_{\mrm{UCB}}} - \nu + 1\rp^{+} \Big| \mcal{F}_{\nu-1} \rb &\overset{(a)}{\leq} \E_{\nu,\mcal{A}}^{A_{\mrm{UCB}}} \lb \lp \tilde{T}_{b, \ell}^{A_{\mrm{UCB}}} - \nu + 1\rp^{+} \Big| \mcal{F}_{\nu-1} \rb \nonumber\\
    &= \E_{\nu,\mcal{A}}^{A_{\mrm{UCB}}} \lb \tilde{T}_{b, \ell}^{A_{\mrm{UCB}}} - \nu + 1 \Big| \mcal{F}_{\nu-1} \rb \nonumber\\
    &= \ell - \nu + \E_{\nu,\mcal{A}}^{A_{\mrm{UCB}}} \lb \tilde{T}_{b, \ell}^{A_{\mrm{UCB}}} - \ell + 1 \Big| \mcal{F}_{\nu-1} \rb \nonumber\\
    &\overset{(b)}{\leq} W + \E_{\nu,\mcal{A}}^{A_{\mrm{UCB}}} \lb \tilde{T}_{b, \ell}^{A_{\mrm{UCB}}} - \ell + 1 \Big| \mcal{F}_{\nu-1} \rb \nonumber\\
    &\overset{(c)}{=} W + \E_{\nu,\mcal{A}}^{A_{\mrm{UCB}}} \lb \tilde{T}_{b, \ell}^{A_{\mrm{UCB}}} - \ell + 1 \rb \nonumber\\
    &\overset{(d)}{=} W + \E_{\ell,\mcal{A}}^{A_{\mrm{UCB}}} \lb \tilde{T}_{b, \ell}^{A_{\mrm{UCB}}} - \ell + 1 \rb \nonumber\\
    &\overset{(e)}{=} W + \E_{1,\mcal{A}}^{A_{\mrm{UCB}}} \lb \tilde{T}_{b, 1}^{A_{\mrm{UCB}}} \rb \nonumber\\
    &\overset{(f)}{=} W + \E_{1,\mcal{A}}^{A_{\mrm{UCB}}} \lb \tilde{T}_{b}^{A_{\mrm{UCB}}} \rb. \label{eq:CuSum_Lem1_laststep_sep}
\end{align}
In step $(a)$, $\tilde{T}_{b, \ell}^{A_{\mrm{UCB}}} \geq \tilde{T}_{b}^{A_{\mrm{UCB}}}$ almost surely, since we can show that $C_{a,\ell:n}^{A_{\mrm{UCB}}} \leq C_{a,n}^{A_{\mrm{UCB}}}$ for $n \geq \ell$ by induction. In step $(b)$, since $\ell$ is the first restart time-step after the change, $\ell - \nu \leq W$. In step $(c)$, because $C_{a,\ell:n}^{A_{\mrm{UCB}}}$ is independent the observations before time-step $\ell$, we can utilize the fact that $\tilde{T}_{b, \ell}^{A_{\mrm{UCB}}}$ is independent of $\mcal{F}_{\ell - 1}$ and $\mcal{F}_{\nu - 1} \subseteq \mcal{F}_{\ell - 1}$. Step $(d)$ follows from the fact that $X_{a,n}$ follows the post-change density $f_{a,1}$ for $n \geq \ell$ and $a \in \mcal{A}$ when the change occurs either at $\nu$ or at $\ell$. Step $(e)$ results from the fact that $X_{a,n}$ follows the post-change density $f_{a,1}$ for $n \geq \ell$ and $a \in \mcal{A}$ under $\mbb{P}_{\ell, \mcal{A}}^{A_{\mrm{UCB}}}$, and that $X_{a,n}$ follows the post-change density $f_{a,1}$ for $n \geq 1$ and $a \in \mcal{A}$ under $\mbb{P}_{1, \mcal{A}}^{A_{\mrm{UCB}}}$. In step $(f)$, $\tilde{T}_{b}^{A_{\mrm{UCB}}}=\tilde{T}_{b, 1}^{A_{\mrm{UCB}}}$  by the definitions in \eqref{eq:CuSum_stat_sep} and \eqref{eq:windowed_CuSum_sep}.
\end{proof}

With Lemma \ref{lem:nu_to_one_PA}, we can derive that
\begin{align}
    \mcal{J}_{\mcal{A}} \lp A_{\mrm{UCB}}, \tilde{T}_{b}^{A_{\mrm{UCB}}} \rp \leq W + \E_{1,\mcal{A}}^{A_{\mrm{UCB}}} \lb \tilde{T}_{b}^{A_{\mrm{UCB}}} \rb.
\end{align}
To upper bound $\E_{1,\mcal{A}}^{A_{\mrm{UCB}}} \lb \tilde{T}_{b}^{A_{\mrm{UCB}}} \rb$, we need the following lemma.

\begin{lemma}\label{lem:E1_upp_PA_UCB_CuSum}
For some constant $c > 0$ that does not depend on $b, I_{\mcal{A}}$, and $W$,
\begin{equation}
    \E_{1,\mcal{A}}^{A_{\mrm{UCB}}} \lb \tilde{T}_{b}^{A_{\mrm{UCB}}} \rb \leq \frac{W \lp b + c \lp 1 + \sqrt{b} \rp \rp}{WI_{\mcal{A}} - I_{\mcal{A}}\lp \sum_{a: \Delta_{a} > 0} 3 + \frac{16 \log W}{\lp\Delta_{a}\rp^{2}} \rp}
\end{equation}
where $\Delta_{a} = I_{\mcal{A}} - D \lp f_{a,1} || f_{a,0} \rp$.
\end{lemma}

\begin{proof}
First, let $a^{*} \coloneqq \argmax_{a \in \lb K \rb} D\lp f_{a,1} || f_{a,0} \rp$ and define the following statistic to be the summation of the LLRs at time-steps when the control policy chooses action $a$: for any $a \in \lb K \rb$ and $n \in \mbb{N}$,
\begin{align}\label{eq:Sum_LLR_sep}
    \Bar{C}_{a,n}^{A_{\mrm{UCB}}}\coloneqq \sum_{i=1}^{n} \mbb{I} \lbp A_{i} = a \rbp \mrm{LLR}\lp a, X_{a, i} \rp.
\end{align}
Also, define $\Bar{T}_{a,b}^{A_{\mrm{UCB}}}$ to be the stopping time at which $\Bar{C}_{a,n}^{A_{\mrm{UCB}}}$ crosses the threshold $b$ at the end of an $W$-length interval $\mcal{I}_{i} = \lbp (i-1)W + 1, \dots, iW \rbp$, i.e., for any $a \in \lb K \rb$ and $n \in \mbb{N}$,
\begin{align}\label{eq:Sum_LLR_stop_sep}
    \bar{L}_{a,b}^{A_{\mrm{UCB}}} \coloneqq \inf \lbp i \in \mbb{N}: \Bar{C}_{a, iW}^{A_{\mrm{UCB}}} \geq b \rbp \mrm{\; and \;} \Bar{T}_{a,b}^{A_{\mrm{UCB}}} \coloneqq \bar{L}_{a,b}^{A_{\mrm{UCB}}}W.
\end{align}
Since $\bar{C}_{a,n}^{A_{\mrm{UCB}}} \leq C_{a,n}^{A_{\mrm{UCB}}}$ for all $n \in \mbb{N}$ and $a \in \lb K \rb$, $\bar{T}_{a^{*},b} \geq \tilde{T}_{b}^{A_{\mrm{UCB}}}$. Recall that $\mcal{G}_{i} \coloneqq \sigma \lp X_{1}, \dots, X_{iW} \rp$ for any $i \in \mbb{N}$ and define the statistic: for any $i \in \mbb{N}$ and $a \in \lb K \rb$,
\begin{equation}
    Y_{a,i}^{A_{\mrm{UCB}}} \coloneqq \Bar{C}_{a,iW}^{A_{\mrm{UCB}}} - \Bar{C}_{a,\lp i - 1 \rp W}^{A_{\mrm{UCB}}} = \sum_{j=(i-1)W + 1}^{iW} \mbb{I} \lbp A_{i} = a \rbp \mrm{LLR}\lp a, X_{a, j}\rp.
\end{equation}
Then,
\begin{align}
    \E_{1,\mcal{A}}^{A_{\mrm{UCB}}} \lb Y_{a^{*},i}^{A_{\mrm{UCB}}} \Big| \mcal{G}_{i-1} \rb &\overset{(a)}{=} \E_{1,\mcal{A}}^{A_{\mrm{UCB}}} \lb Y_{a^{*},i}^{A_{\mrm{UCB}}} \rb \nonumber\\
    &= \E_{1,\mcal{A}}^{A_{\mrm{UCB}}} \lb \sum_{j=(i-1)W + 1}^{iW} \E_{1,\mcal{A}}^{A_{\mrm{UCB}}} \lb \mbb{I} \lbp A_{i} = a^{*} \rbp \mrm{LLR}\lp a^{*}, X_{a^{*}, j}\rp | \mcal{F}_{j-1} \rb \rb \nonumber\\
    &= \E_{1,\mcal{A}}^{A_{\mrm{UCB}}} \lb \sum_{j=(i-1)W + 1}^{iW} \mbb{I} \lbp A_{i} = a^{*} \rbp \E_{1,\mcal{A}}^{A_{\mrm{UCB}}} \lb \mrm{LLR}\lp a^{*}, X_{a^{*}, j}\rp | \mcal{F}_{j-1} \rb \rb \nonumber\\
    &\overset{(b)}{=} \E_{1,\mcal{A}}^{A_{\mrm{UCB}}} \lb \sum_{j=(i-1)W + 1}^{iW} \mbb{I} \lbp A_{i} = a^{*} \rbp \E_{1,\mcal{A}}^{A_{\mrm{UCB}}} \lb \mrm{LLR}\lp a^{*}, X_{a^{*}, j}\rp \rb \rb \nonumber\\
    &\overset{(c)}{\leq} WI_{\mcal{A}},
\end{align}
and
\begin{align}
    \E_{1,\mcal{A}}^{A_{\mrm{UCB}}} \lb Y_{a^{*},i}^{A_{\mrm{UCB}}} | \mcal{G}_{i-1} \rb &= \E_{1,\mcal{A}}^{A_{\mrm{UCB}}} \lb \sum_{j=(i-1)W + 1}^{iW} \mbb{I} \lbp A_{i} = a^{*} \rbp \E_{1,\mcal{A}}^{A_{\mrm{UCB}}} \lb \mrm{LLR}\lp a^{*}, X_{a^{*}, j}\rp \rb \rb \nonumber\\
    &\overset{(d)}{\geq} WI_{\mcal{A}} - I_{\mcal{A}}\sum_{a: \Delta_{a} > 0} \lp 3 +  \frac{16 \log W}{\lp \Delta_{a} \rp^{2}} \rp.
\end{align}
In step $(a)$, we apply the fact that $Y_{a^{*},i}^{A_{\mrm{UCB}}}$ is independent of $\mcal{G}_{i-1}$. Step $(b)$ follows from the fact that $X_{a^{*},j}$ is independent of $\mcal{F}_{j-1}$. In step $(c)$, we apply the fact that $\E_{1,\mcal{A}}^{A_{\mrm{UCB}}} \lb \mrm{LLR}\lp a^{*}, X_{a^{*}, j}\rp \rb = I_{\mcal{A}}$. In step $(d)$, we apply the regret upper bound in Theorem 7.1 in \cite{lattimore2020bandit}, as the LLRs are the rewards of the UCB algorithm. Recall that we assume $\int \lp \log \lp f_{a,1} / f_{a,0} \rp \rp^{2} f_{a,1}d\lambda < \infty$ for any $a \in \lb K \rb$. Then,
\begin{align}
    &\E_{1,\mcal{A}}^{A_{\mrm{UCB}}} \lb \lp Y_{a^{*},i}^{A_{\mrm{UCB}}} \rp^{2} \bigg| \mcal{G}_{i-1} \rb \nonumber\\
    &\overset{(a)}{=} \E_{1,\mcal{A}}^{A_{\mrm{UCB}}} \lb \lp Y_{a^{*},i}^{A_{\mrm{UCB}}} \rp^{2} \rb \nonumber\\
    &= \E_{1,\mcal{A}}^{A_{\mrm{UCB}}} \lb \sum_{j=(i-1)W + 1}^{iW} \sum_{k=(i-1)W + 1}^{iW} \mbb{I} \lbp A_{j} = a^{*} \rbp \mbb{I} \lbp A_{k} = a^{*} \rbp \mrm{LLR}\lp a^{*}, X_{a^{*}, j}\rp \mrm{LLR}\lp a^{*}, X_{a^{*}, k}\rp \rb \nonumber\\
    &\leq \sum_{j=(i-1)W + 1}^{iW} \sum_{k=(i-1)W + 1}^{iW} \E_{1,\mcal{A}}^{A_{\mrm{UCB}}} \lb \mrm{LLR}\lp a^{*}, X_{a^{*}, j}\rp \mrm{LLR}\lp a^{*}, X_{a^{*}, k}\rp  \rb \nonumber\\
    &= \sum_{j=(i-1)W + 1}^{iW} \E_{1,\mcal{A}}^{A_{\mrm{UCB}}} \lb \lp\mrm{LLR}\lp a^{*}, X_{a^{*}, j}\rp\rp^{2} \rb + 2\sum_{j=(i-1)W + 1}^{iW} \sum_{k=j+1}^{iW} \E_{1,\mcal{A}}^{A_{\mrm{UCB}}} \lb \mrm{LLR}\lp a^{*}, X_{a^{*}, j}\rp \mrm{LLR}\lp a^{*}, X_{a^{*}, k}\rp \rb \nonumber \\
    &\overset{(b)}{\leq} Wv + (W-1)WI_{\mcal{A}}^{2}
\end{align}
where step $(a)$ stems from the fact that $Y_{a^{*},i}^{A_{\mrm{UCB}}}$ is independent of $\mcal{G}_{i-1}$, and step $(b)$ follows from Assumption \ref{assum:KL}. Then, we can apply Proposition 1 in \cite{veeravalli2024quickest} and obtain that 
\begin{equation}\label{eq:L_upp_sep}
    \E_{1,\mcal{A}}^{A_{\mrm{UCB}}} \lb L_{a^{*},b}^{A_{\mrm{UCB}}} \rb \leq \frac{b + c \lp 1 + \sqrt{b} \rp}{WI_{\mcal{A}} -  I_{\mcal{A}}\sum_{a: \Delta_{a} > 0} \lp 3 + \frac{16 \log W}{\lp \Delta_{a} \rp^{2}} \rp}.
\end{equation}
Thus, 
\begin{equation}\label{eq:final_upp_sep}
    \E_{1,\mcal{A}}^{A_{\mrm{UCB}}} \lb \tilde{T}_{b}^{A_{\mrm{UCB}}} \rb \overset{(a)}{\leq} \E_{1,\mcal{A}}^{A_{\mrm{UCB}}} \lb \Bar{T}_{a^{*},b}^{A_{\mrm{UCB}}} \rb \overset{(b)}{\leq} \frac{W \lp b + c \lp 1 + \sqrt{b} \rp \rp}{WI_{\mcal{A}} - I_{\mcal{A}}\sum_{a: \Delta_{a} > 0} \lp 3 + \frac{16 \log W}{\lp \Delta_{a} \rp^{2}} \rp}
\end{equation}
where step $(a)$ follows from the fact that $\bar{T}_{a^{*},b}^{A_{\mrm{UCB}}} \geq \tilde{T}_{b}^{A_{\mrm{UCB}}}$, and step $(b)$ results from \eqref{eq:L_upp_sep}.
\end{proof}

With Lemma \ref{lem:E1_upp_PA_UCB_CuSum}, we derive
\begin{equation}
    \mcal{J}_{\mcal{A}} \lp A_{\mrm{UCB}}, \tilde{T}_{b}^{A_{\mrm{UCB}}} \rp \leq W + \frac{W \lp b + c \lp 1 + \sqrt{b} \rp \rp}{WI_{\mcal{A}} - I_{\mcal{A}}\sum_{a: \Delta_{a} > 0} \lp 3 + \frac{16 \log W}{\lp \Delta_{a} \rp^{2}} \rp}.
\end{equation}
Then, as $W, b \to \infty$ with $W = o(b)$,
\begin{equation}
     \mcal{J}_{\mcal{A}} \lp A_{\mrm{UCB}}, \tilde{T}_{b}^{A_{\mrm{UCB}}} \rp \leq \frac{b}{I_{\mcal{A}}} \lp 1 + o \lp 1 \rp \rp.
\end{equation}
This completes the proof of Theorem \ref{thm:PA_UCB_CuSum_opt}.

\end{document}